\documentclass[a4paper,10pt]{article}
\author{Palash Dey$^\star$ and Sourav Medya$^\dagger$\\ \texttt{palash.dey@cse.iitkgp.ac.in, medya@cs.ucsb.edu}\\$^\star$Indian Institute of Technology, Kharagpur\\$^\dagger$University of California, Santa Barbara}

\usepackage{xspace}
\usepackage{xcolor}
\usepackage{color}
\usepackage{fullpage}
\definecolor{antiquebrass}{rgb}{0.8, 0.58, 0.46}
\usepackage[colorlinks=true,linkcolor=blue,citecolor=antiquebrass]{hyperref}
\usepackage{algorithmic}
\usepackage[ruled,vlined,linesnumbered]{algorithm2e}
\usepackage{epsfig}
\usepackage{subfig}
\usepackage{mathrsfs,amsmath,amsthm,amssymb,mdwlist}

\usepackage{enumerate}

\usepackage[shortlabels]{enumitem}

\usepackage{nicefrac}

\newcommand{\nfrac}{\nicefrac}

\newcommand{\eps}{\ensuremath{\varepsilon}\xspace}
\renewcommand{\epsilon}{\eps}
\let\mydelta\delta
\renewcommand{\delta}{\ensuremath{\mydelta}\xspace}
\let\myalpha\alpha
\renewcommand{\alpha}{\ensuremath{\myalpha}\xspace}
\let\mystar\star
\renewcommand{\star}{\ensuremath{\mystar}}

\newcommand{\HL}{{\sc Hiding Leader}\xspace}

\newcommand{\SC}{{\sc Set Cover}\xspace}

\newcommand{\DKS}{{\sc  Densest $k$-Subgraph}\xspace}

\newcommand{\pr}{\ensuremath{\prime}}

\newcommand{\el}{\ensuremath{\ell}\xspace}

\newcommand{\NP}{\ensuremath{\mathsf{NP}}\xspace}
\newcommand{\Pb}{\ensuremath{\mathsf{P}}\xspace}

\newcommand{\NPC}{\ensuremath{\mathsf{NP}}-complete\xspace}

\newcommand{\YES}{{\sc{yes}}\xspace}
\newcommand{\NO}{{\sc{no}}\xspace}

\newcommand{\NB}{\ensuremath{\mathbb N}\xspace}

\newcommand{\CC}{\ensuremath{\mathcal C}\xspace}

\newcommand{\EE}{\ensuremath{\mathcal E}\xspace}
\newcommand{\FF}{\ensuremath{\mathcal F}\xspace}
\newcommand{\GG}{\ensuremath{\mathcal G}\xspace}
\newcommand{\HH}{\ensuremath{\mathcal H}\xspace}

\newcommand{\JJ}{\ensuremath{\mathcal J}\xspace}

\newcommand{\LL}{\ensuremath{\mathcal L}\xspace}

\newcommand{\NN}{\ensuremath{\mathcal N}\xspace}
\newcommand{\OO}{\ensuremath{\mathcal O}\xspace}

\renewcommand{\SS}{\ensuremath{\mathcal S}\xspace}

\newcommand{\UU}{\ensuremath{\mathcal U}\xspace}
\newcommand{\VV}{\ensuremath{\mathcal V}\xspace}
\newcommand{\WW}{\ensuremath{\mathcal W}\xspace}
\newcommand{\XX}{\ensuremath{\mathcal X}\xspace}
\newcommand{\YY}{\ensuremath{\mathcal Y}\xspace}
\newcommand{\ZZ}{\ensuremath{\mathcal Z}\xspace}

\usepackage{balance}

\newtheorem{theorem}{\bf Theorem}
\newtheorem{proposition}[theorem]{\bf Proposition}

\newtheorem{corollary}[theorem]{\bf Corollary}
\newtheorem{definition}[theorem]{\bf Definition}

\usepackage{cleveref}

\crefname{theorem}{theorem}{\bf Theorems}
\crefname{observation}{observation}{\bf Observations}
\crefname{lemma}{lemma}{\bf Lemmas}
\crefname{corollary}{corollary}{\bf Corollaries}
\crefname{proposition}{proposition}{\bf Propositions}
\crefname{definition}{definition}{\bf Definitions}
\crefname{claim}{claim}{\bf Claims}
\crefname{reductionrule}{reduction rule}{\bf Reduction rules}

\title{Covert Networks: How Hard is It to Hide?}

\begin{document}

\maketitle

\begin{abstract}
Covert networks are social networks that often consist of harmful users. Social Network Analysis (SNA) has played an important role in reducing criminal activities (e.g., counter terrorism) via detecting the influential users in such networks. There are various popular measures to quantify how influential or central any vertex is in a network. As expected, strategic and influential miscreants in covert networks would try to hide herself and her partners (called {\em leaders}) from being detected via these measures by introducing new edges.

Waniek et al.~\cite{DBLP:conf/atal/WaniekMRW17} show that the corresponding computational problem, called \HL, is \NPC for the degree and closeness centrality measures. We study the popular core centrality measure and show that the problem is \NPC even when the core centrality of every leader is only $3$. On the contrary, we prove that the problem becomes polynomial time solvable for the degree centrality measure if the degree of every leader is bounded above by any constant. We then focus on the optimization version of the problem and show that the \HL problem admits a $2$ factor approximation algorithm for the degree centrality measure. We complement it by proving that one cannot hope to have any $(2-\eps)$ factor approximation algorithm for any constant $\eps>0$ unless there is a $\nfrac{\eps}{2}$ factor polynomial time algorithm for the \DKS problem which would be considered a significant breakthrough. We empirically establish that our $2$ factor approximation algorithm frequently finds out a near optimal solution. On the contrary, for the core centrality measure, we show that the \HL problem does not admit any $(1-\alpha)\ln n$ factor approximation algorithm for any constant $\alpha\in(0,1)$ unless $\Pb=\NP$ even when the core centrality of every leader is only $3$. Hence, our work shows that, although classical complexity theoretic framework fails to shed any light on relative difficulty of \HL for different centrality measures, the problem is significantly ``harder'' for the core centrality measure than the degree centrality one.
\end{abstract}

\section{Introduction}

Social network analysis (SNA) has played a pivotal role in many applications in multi-agent systems and artificial intelligence~\cite{DBLP:conf/atal/SabaterS02,otte2002social,wang2007social,carrington2005models}. One of the most successful applications of SNA is in counter-terrorism via analyzing {\em covert networks}~\cite{chen2005coplink,ressler2006social,xu2005criminal,lu2010social}. Covert network loosely refers to network of criminals, terrorists, illegal activities, etc. Security personnel regularly use various SNA tools to understand criminal behavior, catch their leaders, and effectively dismantle such networks~\cite{eiselt2018destabilization,farooq2017covert,knoke2015emerging}.

{\em Centrality measure} is one of the most useful tools that SNA provides to analyze covert networks. It assigns scores to the vertices based on their relative \emph{influence or importance} in the network~\cite{bavelas1948mathematical}; depending on the centrality measure, higher scores may correspond to important vertices and important vertices are expected to be more central. One of the simplest and oldest such centrality measures is the {\em degree centrality} which ranks vertices according to their degree~\cite{shaw1954group}. Other important examples include {\em closeness centrality} and {\em betweenness centrality} that are measures based on shortest paths~\cite{beauchamp1965improved}. %which ranks the vertices based on their average distance to every other vertices in the network~\cite{beauchamp1965improved},  which ranks based on the fractional number of shortest paths that pass through them. 
Another centrality measure is the {\em core centrality}~\cite{seidman1983network} which ranks the vertices based on their {\em core number}. Intuitively speaking, if a vertex has a high core number, then it is part of some dense cohesive community within the network. Formally, a {\em $k$-core} is an induced subgraph of the network where the minimum degree of the vertices is at least $k$. The core number of a vertex is the highest integer $k$ such that the vertex is part of some {\em $k$-core}. Therefore, the core centrality can be more revealing about the position of a node than its degree centrality---while degree centrality only concerns about the degree of a vertex, the core centrality elegantly takes into consideration the degrees of the neighbors as well as the vertex. These two measures are also related in the sense that the core centrality of any vertex is at most its degree centrality. Due to its sophisticated nature, the core centrality has been extensively used in the study of covert networks~\cite{morselli2007efficiency,shaikh2008network,memon2012identifying} as well as in other important tasks such as viral marketing and social engagement~\cite{kitsak2010identification,bhawalkar2015preventing} in social networks.

% As Waniek et al.~\cite{DBLP:conf/atal/WaniekMRW17} argues, in spite of having plenty of SNA tools, understanding and exploiting terrorist organizational network has remained challenging due to various reasons~\cite{krebs2002mapping,sparrow1991application,sageman2004understanding,alimi2015dynamics,roberts2016monitoring,baker1993social} ranging from incomplete data, erroneous data, to level of friendship for different links.

In this paper, our goal is to study the centrality measure based secrecy in covert networks. Indeed, understanding covert networks remains a challenging task mainly due to incompleteness and dynamic evolution of the data as well as the strategic nature of the users ~\cite{krebs2002mapping,sparrow1991application,sageman2004understanding,alimi2015dynamics,roberts2016monitoring,baker1993social}. Since the criminals often possess technical expertise~\cite{johnson2016new,stevenson2014change,crossley2012covert,calvey2017covert,atran2017challenges}, we are interested in the evolution of terrorist networks under a framework of strategic users~\cite{DBLP:conf/atal/WaniekMRW17}: \textit{How is the network designed to hide the central or influential users aka the leaders?}   %Although most research in SNA makes the standard assumption that the entities corresponding to vertices do not act strategically to evade detection using various SNA tools, as Waniek et al.~\cite{DBLP:conf/atal/WaniekMRW17} argue, it may be foolish to make such an assumption for covert networks since criminals often possess technical expertise~\cite{johnson2016new,stevenson2014change,crossley2012covert,calvey2017covert,atran2017challenges}. 

Waniek et al.~\cite{DBLP:conf/atal/WaniekMRW17} first propose the \HL problem which incorporates the viewpoint of the leaders of a criminal organization. It also explicitly models knowledge of the criminals about SNA tools that are used to detect them and thus help in dismantling their organization. Intuitively, the input in the \HL problem is a network with a subset of vertices marked as leaders. The goal is to add fewest edges to ensure that various SNA tools do not rank any leader high based on centrality measures thereby capturing the {\em efficiency vs secrecy dilemma} that the criminals are believed to possess~\cite{morselli2007efficiency,csermely2013structure,von2015organized}. Waniek et al. show promising results that the \HL problem is computationally intractable even for the simplest degree centrality measure.% they also prove intractability for the closeness centrality.
% 
% In this paper we explore the secrecy about the identity of the important members aka \textit{the leaders}. In particular we are interested in the \HL problem: \textit{how should the leaders manipulate the network such that they can hide their identity by any centrality measure~\cite{DBLP:conf/atal/WaniekMRW17}?} 
% 
% 
% % centrality
% Our model explores the opportunity of the leaders being hidden in covert networks. It also keeps the influence of these leaders same or larger in the process of hiding. We follow the model proposed by Waniek et al.~\cite{DBLP:conf/atal/WaniekMRW17}, where the leaders are trying to shield themselves from law enforcement agencies via centrality measures. However, our framework of evaluating secrecy includes a new centrality measure---\textit{core centrality}. A $k$-core \cite{seidman1983network} is a maximal induced subgraph where the vertices have degree at least $k$. This structure is useful to model users' engagement as well as viral marketing in social networks \cite{bhawalkar2015preventing,kitsak2010identification}. If a node is in $k$-core, then its core centrality is $k$. This centrality certainly captures a higher order perspective than degree centrality as it involves connections between the nighbours.  

% past work on analysis

% algorithmic perspective

\subsection{Contribution}

In this paper, we study the \HL problem for the core centrality measure and show that the degree centrality measure is much more computationally vulnerable than the core centrality measure although the \HL problem is \NPC for both of them. We reinforce our above claim further through extensive empirical evaluations. Our specific contribution in this paper are as follows.

\begin{itemize}[itemsep=.2cm]
 \item We show that the \HL problem for degree centrality is polynomial time solvable if the degree of every leader is bounded by some constant [\Cref{thm:deg_poly}].
 
 \item We present a $2$ factor approximation algorithm for the \HL problem for degree centrality which optimizes the number of edges added [\Cref{thm:2apx_opt_fol_deg}]. We complement this by proving that, if there exists a $(2-\eps)$ factor approximation algorithm for the above problem for any constant $0<\eps<1$, then there exists a $\nfrac{\eps}{2}$ factor approximation algorithm for the \DKS problem [\Cref{thm:2apx_hard}] which would be considered a substantial breakthrough. To the best of our knowledge, the state of the art algorithm for the \DKS problem achieves an approximation ratio of $\tilde{\OO}(n^{\nfrac{1}{4}})$ only~\cite{DBLP:conf/soda/BhaskaraCVGZ12}.
 
 \item For the core centrality measure, we show that the \HL problem is \NPC even if the core centrality of every leader is exactly $3$ [\Cref{thm:CHL_hard}]. We prove that our result is almost tight in the sense that the \HL problem is polynomial time solvable if the core centrality of every leader is at most $1$ [\Cref{thm:core_poly}]. Moreover, we also prove that there does not exist any $(1-\alpha)\ln n$ factor approximation algorithm for any constant $\alpha\in(0,1)$ which optimizes the number of edges that one needs to add even when the core centrality of every leader is $3$ [\Cref{cor:core_inapp}].
 
 \item We show that a construction of a network by Waniek et al.~\cite{DBLP:conf/atal/WaniekMRW17}, called ``captain network'' there, hides the leaders with respect to the core centrality measure also.
 
 \item We empirically evaluate our $2$-approximation algorithm for the degree centrality measure in synthetic networks. We observe that our algorithm almost always produces near optimal results in practice. In the experimental results, we also show the extent in which a leader can hide in the captain network with respect to core centrality.
\end{itemize}

\subsection{Related Work}
Waniek et al. first proposed and studied the \HL problem~\cite{DBLP:conf/atal/WaniekMRW17,waniek2018hiding}. They proved that the problem is \NPC for both the degree and closeness centrality measures. They also proposed a procedure to design a captain (covert) network from scratch which not only hides the leaders based on the degree, closeness, and betweenness centrality measures, but also keeps the influence of the leaders high in the network. In this paper, we provide two approximability results for degree centrality and core centrality respectively. We also show the problem is harder in the case of core centrality. Liu et al. \cite{liu2008towards} studied another related problem to make the degree of each node in the network beyond a given constant by adding minimal edges. 

Other problems that align with privacy issues in social networks were studied before \cite{Zheleva2009,altshuler2013stealing}. In \cite{Zheleva2009}, the authors showed how an adversary exploits online social networks to find the private information about users. Altshuler et al. \cite{altshuler2013stealing} discussed the threat of malware targeted at extracting information  in a real-world social network.

\textbf{Computing centrality and related problems. }A significant amount of related work study the computationally complexity of various centrality measures. Brandes~\cite{brandes2001} first proposed an efficient algorithm to compute the betweenness centrality of a vertex in a network. More recently, Riondato et al. \cite{riondato2014} introduced an approach to compute the top-$k$ vertices according to the betweenness centrality using VC-dimension theory. Yoshida \cite{yoshida2014} studied similar problems for both the betweenness and coverage centrality measures in a group setting. Mahmoody et al. subsequently improved the performance of the above algorithms using a novel sampling scheme~\cite{mahmoodye2016}. There is an active line of research to optimize the centrality of one node as well as of a set of nodes~\cite{crescenzi2015,ishakian2012framework,DANGELO2016153,medya2018group}. Nikos et al. proposed a novel procedure to maximize the expected decrease in shortest path distances from a given node to the remaining nodes via edge addition~\cite{parotsidis2016centrality}. Crescenzi et al.~\cite{crescenzi2015} proposed greedy algorithms to increase centrality of certain vertices and show effectiveness of their approach through extensive simulation. Kilberg~\cite{kilberg2012basic} and others studied behavioral models to understand why certain network topologies are common in covert networks~\cite{crossley2012covert,demiroz2012anatomy,belli2015exploring}. Enders and Su~\cite{enders2007rational} and others develop models to explain various properties like efficiency vs secrecy dilemma etc. of covert networks~\cite{janssen2012stable,lindelauf2009influence,duijn2014relative,enders2010network}. Other important direction includes quantifying the influence of vertices; most prominent among them include {\em Independent Cascade} model~\cite{goldenberg2001using}, {\em Linear Threshold} model~\cite{kempe2003maximizing}, Bass model~\cite{bass1969new,meade2006modelling}, etc.

\textbf{Other network design problems. }
We also provide a few details about previous work on other network modification (design) problems. A set of design problems were introduced in \cite{paik1995}. Lin et al.~\cite{lin2015} addressed a shortest path optimization problem via improving edge weights on undirected graphs. The node version of this problem was also studied \cite{dilkina2011,medya2018noticeable,medya2018making}. Meyerson et al.~\cite{meyerson2009} proposed approximation algorithms for single-source and all-pair shortest paths minimization. Faster algorithms for some of these problems were also presented in~\cite{papagelis2011,parotisidis2015selecting}.
Demaine et al.~\cite{demaine2010} minimized the diameter of a network by adding shortcut edges. Dey et al.~\cite{DBLP:conf/comad/DeyKN19} studied the social network effect in the surprise in elections.

\section{Preliminaries}

For a positive integer \el, we denote the set $\{1,2,\ldots,\el\}$ by $[\el]$. A network or graph $\GG=(\VV,\EE)$ is a tuple consisting of a finite set $\VV$ (or $\VV[\GG]$) of $n$ vertices and a set $\EE\subseteq\VV\times\VV$ of edges (also denoted by $\EE[\GG]$). A network is called {\em undirected} if we have $(x,y)\in\EE$ whenever we have $(y,x)\in\EE$ for any $x,y\in\VV$ with $x\ne y$. A {\em self loop} is an edge of the form $(x,x)$ for some $x\in\VV$. In this paper, we focus on undirected networks without any self loop. The {\em degree} of a vertex $x$ is the number of edges incident on it which is $|\{e\in\EE: x\in e\}|$. A {\em subgraph} of a network $\GG=(\VV,\EE)$ is a network $\HH=(\UU,\FF)$ such that $\UU\subseteq\VV$ and $\FF\subseteq \EE\cap(\UU\times\UU)$. For a positive integer $k$, a subgraph \HH of a network \GG is called a {\em $k$-core} if the degree of every vertex in \HH at least $k$. The {\em core number} of a vertex $x$ in a network is the largest integer $k$ such that $x$ belongs to a $k$-core.

\subsection{Network Centrality}
Let \GG be any network. Bavelas~\cite{bavelas1948mathematical} introduces the notion of centrality of vertices. Intuitively, centrality measures try to capture the importance of a vertex in a network. Shaw~\cite{shaw1954group} proposes the {\em degree centrality} measure which has turned out to be one of the most useful measures. The degree centrality of a vertex $x$ in \GG is the degree $\text{deg}_\GG(x)$ of $x$ in the network, that is $|\{y\in\VV[\GG]:\{x,y\}\in\EE[\GG]\}|$. 

Seidman~\cite{seidman1983network} introduces the idea of {\em core centrality} which is particularly useful for finding network cohesion. For an integer $k$, a {\em $k$-core} is a subgraph \HH of \GG such that the degree of every vertex in \HH is at least $k$. The core number of a vertex $x$ in \GG is the largest $k$ such that $x$ belongs to a $k$-core, that is $\max\{k\in\NB:\exists \HH\subseteq\GG, \text{deg}_\HH(x)\ge k \forall x\in\VV[\HH]\}$. The core centrality of a vertex $x$ in \GG is its core number in the network. Other popular network centrality measures includes closeness centrality~\cite{beauchamp1965improved}, betweenness centrality~\cite{anthonisse1971rush,freeman1977set}, etc.
% 
% \subsection{Influence Model in Networks}
% 
% Intuitively, a network influence model assigns scores to pairs of vertices $(x,y)\in\VV[\GG]\times\VV[\GG]$ of the network based on how much influence the vertex $x$ has on the vertex $y$ the network. One of the popular such models is the {\em Independent Cascade (IC)} model of influence~\cite{kempe2003maximizing}. In this model, the edges of the network are assigned an activation probability $p:E\rightarrow [0,1]$. To begin with (in round $t=0$), there is a subset of vertices which are {\em active}; they are called {\em seed vertices}. In each round $t\ge 1$, an active vertex that has become active in ($t-1$)-th round activates all of its inactive neighbours independently with corresponding edge probabilities between them. The process ends when there is no vertex to activate further. In the IC model, the influence of a vertex $x$ to another vertex $y$ is defined as the {\em probability that the vertex $y$ gets activated given the set of seed vertices to be $\{x\}$}. Let $\QQ(x,y)$ be the set of paths between $x$ and $y$ in the network. Then the influence of $x$ on $y$ in the network under IC model is defined as follows.
% \begin{equation}
%  \text{inf}_{\text{IC}} (y;x) = \sum_{q\in \QQ(x,y)} \prod_{e\in q}p_{e}    \label{eqn:influence}
% \end{equation}

\subsection{Problem Definition}

Intuitively, the input in the \HL problem is a network with a subset of vertices marked as leaders (and the other vertices are followers), a budget $b$ which is the maximum number of edges that we can add in the network, and a target $d$ which is the minimum number of followers whose centrality must be at least as high as the centrality of any leader in the resulting network (after addition of the new edges). We now define our problem formally. In \Cref{probdef}, $c(\cdot,\cdot)$ denote either degree centrality or core centrality.\vspace{1ex}

 \noindent\fbox{\begin{minipage}{\linewidth}
\begin{definition}[\HL (HL)]\label{probdef}
Given a graph $\GG=(\VV,\EE)$, a subset $\LL\subseteq\VV$ of leader vertices, an integer $b$ denoting the maximum number of edges that we are allowed to add in \GG, an integer $d$ denoting the number of follower vertices in $\FF=\VV\setminus\LL$ whose final centrality should be at least as high as any leader, the goal is to compute if there exists a subset $\WW\subseteq\FF\times\FF$ of edges between followers such that the conditions below hold.
\begin{enumerate}[(i)]
 \item $|\WW|\le b$
 \item $\exists_{\FF^\prime \subseteq \FF} |\FF^\prime|\geq d $ such that 
 $$c(\GG^\prime,f) \geq c(\GG^\prime,l), \forall {f\in \FF^\prime,l\in \LL}$$
 where $\GG^\prime=(\VV,\EE\cup\WW)$.
\end{enumerate}
\end{definition}
\end{minipage}}

\section{Results for Degree Centrality}

We present our algorithmic and hardness results for the \HL problem for the degree centrality measure in this section. We begin with presenting our polynomial time algorithm for the \HL problem for the degree centrality measure when the degree of every leader in the network is bounded above by any constant. On a high level, our algorithm makes greedy choices as long as it can and uses local search technique when ``stuck.''

\begin{theorem}\label{thm:deg_poly}
 There exists a polynomial time algorithm for the \HL problem for degree centrality if the degree of every leader is bounded by any constant.
\end{theorem}

\begin{proof}
 Let \GG be the input graph and $k$ the highest degree of any leader; that is $k=\max\{\text{deg}_\GG(l): l\in L\}$. We are given that $k$ is a constant. If the number of followers is at most $2k$, then there are at most ${2k\choose 2}$ (which is a constant) new edges that we can add and we try all possible subsets of it of cardinality at most $b$. The number of such subsets is at most $2^{{2k\choose 2}}$ which is a constant and thus we can output correctly in polynomial time. So let us assume that the number of followers is at least $2k+1$. Similarly we can also assume without loss of generality that the budget $b$ is at least $4k^2$ since otherwise we will try to add all possible $b$ new edges (there are only $\OO(n^{2b})=n^{\OO(1)}$ possibilities since $k=\OO(1)$) and thus we can output correctly in polynomial time.
 
 Suppose there are already $d^\pr$ number of followers in \GG whose degrees are at least $k$. If $d^\pr\ge d$, we output \YES. Otherwise let us assume without loss of generality that $d^\pr<k$. Let $\XX\subseteq \FF$ with $|\XX|=d-d^\pr$ be the set of top $d-d^\pr$ highest degree followers in the network whose degrees are less than $k$. Intuitively, our algorithm greedily adds new edges between two vertices in \XX whose degrees are less than $k$ until it is stuck and removes some edge it had added before to make progress. Concretely, our algorithm works as follows. To distinguish existing (old) edges from newly added edges (by the algorithm) in \GG, we color the existing (old) edges as red and whenever we add a new edge, we color it green. To begin with, all the edges in \GG are colored red and there is no green edge. We apply the following step (\star) as long as we can. If the number of green edges in \GG is less than $b$ and there exist two vertices $x,y\in\XX$ such that the degrees of both the vertices are less than $k$ and there is no edge between them, then we add an edge $\{x,y\}$ in \GG and color it green. Such a pair of vertices (if exists) can be found in $\OO(n^2)$. If such a pair of vertices does not exist in \GG, then one of the following four cases must hold.
 
 {\em Case $1$: The degree of every vertex in \XX is at least $k$.} In this case, we output \YES. 
 
 {\em Case $2$: The number of green edges in \GG is $b$.} In this case, we output \YES if the degree of every vertex in \XX is at least $k$; otherwise we output \NO. 
 
 {\em Case $3$: There exists exactly one vertex $x\in\XX$ with degree less than $k$.}  If the degree of $x$ is $k-1$, then we add an edge between $x$ and any vertex $y\in L$ such that there is no edge between $x$ and $y$ in \GG and color it green. If the number of green edges in \GG is at most $b$, then we output \YES; otherwise we output \NO. Otherwise we assume that the degree of $x$ is less than $k-1$. If the number of green edges in \GG is at most $3k^2$, then we can add $k$ new green edges on $x$ and answer \YES since we have already assumed that $b\ge 4k^2$. So let us assume without loss of generality that the number of green edges in \GG is more than $3k^2$. Let $\{u,v\}$ be a green edge such that there is no edge between $x$ and $u$ and between $x$ and $v$ in \GG. Such a green edge $\{u,v\}$ always exists in \GG since the degree of $x$ is less than $k-1$ in \GG and there are more than $3k^2$ green edges in \GG. Moreover such an edge can be found in polynomial time by simply checking all the green edges. We now remove the green edge $\{u,v\}$ from \GG, add two edges $\{x,u\}$ and $\{x,v\}$ in \GG, color both of them green, and continue (return to the step (\star)).
 
 {\em Case $4$: For every pair of vertices $x,y\in\XX, x\ne y$ with degree less than $k$ for both the vertices, there is an edge between them.} Let $\ZZ\subseteq\XX$ be the set of vertices in \XX with degree less than $k$. Since there exists an edge between every pair of vertices in \ZZ in this case and the degree of every vertex in \ZZ is less than $k$, we have $|\ZZ|\le k$. If the number of green edges in \GG is at most $3k^2$, then we can add $k$ new green edges on every vertex in \ZZ and answer \YES since we have $|\ZZ|\le k$ and $b\ge 4k^2$. So let us assume that the number of green edges in \GG is more than $3k^2$. Let $a,b\in\XX$ be any two vertices. Let $\NN(a)$ and $\NN(b)$ denote the set of neighbors of $a$ and $b$ in \XX. Since the degrees of both $a$ and $b$ are less than $k$, we have $|\NN(a)|<k$ and $|\NN(b)|<k$. Since the degree of every vertex in \XX is at most $k$ and the number of green edges in \GG is at least $3k^2$, there exists at least one green edge $\{u,v\}$ in \GG which does not incident on any vertex in $\NN(a)\cup\NN(b)\cup\{a,b\}$ (there can be at most $2k^2$ green edges incident on any vertex in this set). Moreover such an edge can be found in polynomial time by simply checking all the green edges. We now remove the green edge $\{u,v\}$ from \GG, add two edges $\{a,u\}$ and $\{b,v\}$ in \GG, color both of them green, and continue (return to the step (\star)).
 
 The algorithm always terminates in polynomial time since in every iteration, it adds a green edge and at most $b$ $(<n^2)$ green edges could be added -- in cases $3$ and $4$, we have added two green edges and removed only one green edge; the algorithm terminates in cases $1$ and $2$. Also, whenever the algorithm outputs \YES, adding green edges to the graph makes the degree of at least $d$ followers in the network at least $k$. Hence, if the algorithm outputs \YES, the instance is indeed a \YES instance. So, let us assume that the algorithm outputs \NO. Except (in case $3$) when there exists exactly one vertex $x$ in the network with degree less than $k$ and the degree of $x$ is $k-1$, whenever we add one green edge in total (which is the same as adding two green edges and removing one green edge in cases $3$ and $4$), the sum of the degrees of all the vertices in \XX increases by $2$. Hence the number of green edges added in the graph is at most $\text{ALG}=\lceil\sum_{x\in\XX}\nfrac{(k-deg(x))}{2}\rceil$. Since the algorithm outputs \NO, we have $\text{ALG}>b$. We observe that since any edge increases the degree of at most $2$ vertices, when the algorithm outputs \NO, the instance is indeed a \NO instance. Hence the algorithm is correct.
\end{proof}

We now present a simple $2$ factor polynomial time approximation algorithm for the \HL problem for degree centrality.

\begin{theorem}\label{thm:2apx_opt_fol_deg}
 There exists a polynomial time algorithm (HLDA) for approximating the budget $b$ in \HL within a factor of $2$ for degree centrality.
\end{theorem}

\begin{proof}
 Let $\FF^\pr\subseteq \FF$ be the set of followers in $\FF$ whose degree centrality is at least the degree centrality of every vertex in $\LL$. Let $|\FF^\pr|=d^\pr$. If $d^\pr\ge d$, then we output an empty set of edges. Let $x_i\in \FF, i\in[d-d^\pr]$ be the $(d-d^\pr)$ followers with highest degree centrality among the vertices in $\FF\setminus \FF^\pr$. We keep on adding edges with at least one end point in $\{x_i: i\in[d-d^\pr]\}$ until the degree centrality of every $x_i, i\in[d-d^\pr]$ is at least the degree of every vertex in $\LL$ in the resulting graph. When we have $d$ followers with degree at least the degree of every vertex in \LL, we output the set of edges that we have added. The algorithm adds at most $\sum_{i=1}^{d-d^\pr} (d-\text{deg}(x_i))$ many edges where where deg($x_i$) is the degree of the vertex $x_i$ in the input graph. Since any new edge can increase the sum of the degrees of the followers by at most $2$, we have OPT$\ge \sum_{i=1}^{d-d^\pr} \nfrac{(d-\text{deg}(x_i))}{2}\ge\nfrac{\text{ALG}}{2}$ by the choice of $\FF^\pr$. Hence our algorithm approximates $b$ by a factor of $2$.
\end{proof}

% \subsubsection*{Extreme Instances}
% We show two instances where the above procedure is optimal and tight $2$-approx respectively.
% 
% \textbf{Optimal:} The above procedure is optimal in the following setting. Let there already exists a set $F^\pr\subseteq F$ of $d^\pr$ followers in $F$ whose degree centrality is more than every vertex in $L$. Let the induced subgraph formed by the vertices in $F\setminus F^\pr$ is a complete graph. Now the addition of an edge (to find the next $d-d'$ followers) by an optimal algorithm would increase the degree of any vertices in  $F\setminus F^\pr$ by at most one. \\
% 
% 
% \textbf{Tight $2$-approximation:} Consider the following construction. The nodes in $L$ have the maximum degree $d-1$. Let us assume $d$ is even. There is a set of followers $F'$ which is clique and has exactly $d$ nodes. That means any node $v\in F'$ has degree $d-1$. There is another set of followers $G'$ is a bipartite graph where one side $G'_L$ has $d$ nodes all with degree $d-1$. The other side $G'_R$ has nodes all with degree $<d-1$. Now if the above mentioned procedure chooses set $F'$ initially, it will end up using $d$ edges to bring all of their degree centrality to $d$. Whereas, an optimal algorithm can add edges inside $G'_R$ and their degree centrality to $d$ by using only $d/2$ edges. \\
% 

% \subsection{Hardness of Approximation}

We now complement our approximation algorithm in \Cref{thm:2apx_opt_fol_deg} by proving that if there exists a polynomial time approximation algorithm for the \HL problem for degree centrality with approximation factor $(2-\eps)$ for any constant $\eps>0$, then there exists a constant factor polynomial time approximation algorithm for the \DKS problem. In the \DKS problem, the input is a graph \GG and an integer $k$ and we need to find a subgraph \HH of \GG on $k$ vertices with highest density. The density of a graph on $n$ vertices is the number of edges in it divided by ${n\choose 2}$. To the best of our knowledge, we do not know whether there exists any polynomial time algorithm which can distinguish a graph containing a clique of size $k$ from a graph where the density of every sub-graph of size $k$ is at most $(\nfrac{\eps}{2})$ (any $\nfrac{\eps}{2}$ factor approximation algorithm for the \DKS problem would be able to distinguish). In fact, none of the known algorithms can distinguish even for some sub-constant values for $\eps$ (see \cite{DBLP:conf/soda/BravermanKRW17} and references therein). We now show that if there exists a $(2-\eps)$ factor approximation algorithm for the \HL problem for any constant $0<\eps<1$, then there exists an $\nfrac{\eps}{2}$ factor approximation algorithm with the same running time (of the \HL algorithm).

\begin{theorem}\label{thm:2apx_hard}
 Suppose there exists a $(2-\eps)$ factor polynomial time approximation algorithm for the \HL problem for degree centrality for some constant \eps. Then there exists a polynomial time algorithm for distinguishing a graph containing a clique of size $k$ from a graph where the density of every sub-graph of size $k$ is at most $\nfrac{\eps}{2}$.
\end{theorem}

\begin{proof}
 Let \GG be any graph which satisfies either (exactly) one of the following properties.
 
 \begin{enumerate}[(i),itemsep=.2cm]
  \item {\bf Completeness:} There exists a clique of size $k$ in \GG.
  \item {\bf Soundness:} The density of any subgraph of \GG of size $k$ is at most $\nfrac{\eps}{2}$.
 \end{enumerate}
 
 From \GG we construct an instance of \HL. Intuitively, we introduce a vertex $a_v$ corresponding to every vertex $v\in\VV[\GG]$ in \GG and the edge set among those vertices is the complement of the corresponding edge set in \GG. To ensure that, in the resulting graph, the degree of every vertex $a_v$ for $v\in\VV[\GG]$ is $n$, we add appropriate number of edges between $a_v$ and some auxiliary vertices $d_{(v,\el)}$ for every $v\in\VV[\GG], \el\in[n]$; we ensure that the degree of any such auxiliary vertex is at most $1$ which will guarantee that these auxiliary vertices are never part of any optimal solution. Finally we add a clique on a set $\{x_i: i\in[n+k]\}$ of leader vertices so that the degree centrality of every leader is $n+k-1$. Formally, the instance $(\HH,\LL,d)$ of \HL is defined as follows.
 \begin{eqnarray*}
  \VV[\HH] &=& \{a_v, d_{(v,\el)}: v\in\VV[\GG], \el\in[n]\} \cup \{x_i: i\in[n+k]\}\\
  \LL &=& \{x_i: i\in[n+k]\}\\
  \EE[\HH] &=& \{\{a_u,a_v\}: (u,v)\notin\EE[\GG]\}\\
  && \cup \{\{x_i,x_j\}:1\le i<j\le n+k\}\\
  && \cup \{\{a_v,d_{(v,\el)}\}: v\in\VV[\GG], \el\in[\text{deg}_\GG(v)+1]\}\\
  d &=& k
 \end{eqnarray*}
 
 We now use the $(2-\eps)$ factor approximation algorithm for the \HL problem for degree centrality which outputs that there is a way to add $b_{ALG}$ number of edges so that there exist at least $d$ followers in \HH whose degree is at least the degree of any leader. Let $b_{OPT}$ denotes the minimum number of edges that one needs to add to ensure that there exist at least $d$ followers in \HH whose degrees are at least the degree of every leader. Then we have $b_{ALG}\le (2-\eps)b_{OPT}$. We output that the graph \GG contains a $k$-clique if $b_{ALG}\le(2-\eps) {k\choose 2}$. Otherwise, we output that the density of any subgraph of \GG on $k$ vertices is at most $\nfrac{\eps}{2}$. We now prove correctness of our algorithm. We first observe that the degree of every leader in \HH is $n+k-1$, the degree of $a_v\in\VV[\HH]$ for every $v\in\VV[\GG]$ is $n$, and the degree of every other vertex is at most $1$.
 
 \begin{enumerate}[(i),itemsep=.2cm]
  \item {\bf Completeness:} Let $\WW\subseteq\VV[\GG]$ with $|\WW|=k$ be a clique in \GG. Let us consider the subset $\XX=\{a_v: v\in\WW\}\subseteq \VV[\HH]\setminus\LL$. By construction, \XX forms an independent set in \HH and the degree of every vertex in \HH is $n$. Since, adding all the edges in $\{\{a_u,a_v\}: u,v\in\WW, u\ne v\}$ in \HH makes the degree of every vertex in \XX in the resulting graph $n+k-1$, we have $b_{OPT}\le {k\choose 2}$. Hence, we have $b_{ALG} \le (2-\eps)b_{OPT}\le (2-\eps) {k\choose 2}$.

  \item {\bf Soundness:} In this case, the density of any subgraph of \GG on $k$ vertices is at most $\nfrac{\eps}{2}$. Let $\YY\subseteq \VV[\HH]\setminus\LL$ with $|\YY|=k$ be a set of any $k$ followers. By the construction of \HH, we have $|\EE[\YY[\HH]]|\ge (1-(\nfrac{\eps}{2})){k\choose 2}$. Hence, the minimum number of edges one needs to add to make the degree of every vertex in \YY at least $n+k-1$ is at least $k(k-1)-(\nfrac{\eps}{2}){k\choose 2} = (2-(\nfrac{\eps}{2})){k\choose 2}$. In particular, we have $b_{ALG}\ge b_{OPT}\ge (2-(\nfrac{\eps}{2})){k\choose 2} > (2-\eps) {k\choose 2}$.
 \end{enumerate}
 
 This concludes the proof of the statement.
\end{proof}

\section{Results for Core Centrality}

We present our results for the \HL problem for the core centrality measure in this section. Unlike in degree centrality case, the problem becomes \NP-complete even when the core centrality of every leader is only $3$. This is almost tight as we prove that the problem is polynomial time solvable if the core centrality of every leader is at most $1$.

In \Cref{thm:CHL_hard} below, we prove that the \HL problem is \NPC even when the core centrality of every leader is $3$. We reduce the \SC problem to the \HL problem there. In the \SC problem, the input is a universe $\UU=\{ u_{1},u_{2},...,u_{n} \}$, a collection $\SS=\{S_{1},S_{2},...,S_{m}\}$ of subsets of \UU, and an integer $t$ and we need to compute if there exist at most $t$ sets in \SS, union of which results in \UU. It is well known that the \SC problem is \NPC~\cite{garey1979computers}.
%We present our results for the \HL problem for core centrality in this section. In \Cref{thm:CHL_hard} below, we prove that the \HL problem is \NPC even when the core centrality of every leader is $3$. We reduce the \SC problem to the \HL problem in \Cref{thm:CHL_hard}. In the \SC problem, the input is a universe $\UU=\{ u_{1},u_{2},...,u_{n} \}$, a collection $\SS=\{S_{1},S_{2},...,S_{m}\}$ of subsets of \UU, and an integer $t$ and we need to compute if there exists at most $t$ sets in \SS whose union is \UU. It is well known that the \SC problem is \NPC~\cite{garey1979computers}.

\begin{theorem} \label{thm:CHL_hard}
 The \HL problem for the core centrality measure is \NPC even when the core centrality of every leader is $3$.
\end{theorem}

\begin{proof}
 The \HL problem for the core centrality measure is clearly in \NP. Note that computing core centrality of a node takes polynomial time \cite{bhawalkar2015preventing}. To prove \NP-hardness, we reduce from the \SC problem. Let $(\UU=\{ u_{1},u_{2},...,u_{n} \},\SS=\{S_{1},S_{2},...,S_{m}\}, t)$ be an instance of the \SC problem. To define a corresponding \HL problem instance, we construct the graph $\GG$ as follows. 
 
 Intuitively, for each subset $S_i\in \SS$, we create a path of $n$ vertices $X_{i,1}, X_{i,2},\cdots, X_{i,n}$ in $\GG$; $ (X_{i,2},X_{i,3}),\cdots,(X_{i,n-1},X_{i,n}), (X_{i,n},X_{i,1})$ are the edges of the above path. We also add $5$ vertices $W_{i,1}$ to $W_{i,5}$ with eight edges where the four vertices in $\{W_{i,\el}:2\le\el\le5\}$ form a clique with six edges; the other two edges are $(W_{i,1},W_{i,2})$ and $(W_{i,1},W_{i,5})$. For each $u_j \in \UU$, we add a set of $5$ vertices $\{Z_{j,\el}:1\le\el\le5\}$ with eight edges where the four vertices (leaders) $\{Z_{i,\el}:2\le\el\le5\}$ form a clique with six edges; the other two edges are $(Z_{j,1},Z_{j,2})$ and $(Z_{j,1},Z_{j,5})$. We also have an edge $(X_{i,j},Z_{j,1})$ for every $u_j\in S_i$. We allow to add $t$ new edges and demand that the core centrality of at least $4m+n(t+1)+t$ followers should be at least as high as the core centrality of every leader. \Cref{fig:hardness_ex1} illustrates the structure of our construction for sets $S_1=\{u_1,u_2\},S_2=\{u_2\},S_3=\{u_3,u_4\}$. We now formally describe our \HL instance.
 \begin{align*}
     V[\GG] &= \{X_{i,j} : S_i\in \SS, u_j\in \UU\} \cup V_1 \\
      \EE[\GG] &= E_1\cup E_2\cup E_3 \cup E_4 \cup E_5\cup E_6\cup E_7\cup E_8\\
      V_1 &= \{W_{i,p}: S_i\in \SS, p\in[5]\} \cup \{Z_{j,p} : u_j\in \UU,  p\in[5]\} \\
      E_1 &= \{(X_{i,j},X_{i,j+1}) : j\in [n-1]\setminus\{1\}, i\in [m]\} \\
       E_2 &= \{(X_{i,n},X_{i,1}) : i\in [m]\}   \\
      E_3 &= \{(W_{i,1}, W_{i,p})|i\in [m], p=2,5\} \\
      E_4 &=  \{(W_{i,1}, X_{i,j})|u_j\notin S_i, i \in [m], j\in [n]\} \\
      E_5 &= \{(Z_{j,1}, Z_{j,p})|j\in [n], p=2,5\}\\
      E_6 &= \{(X_{i,j}, Z_{j,1})|u_j\in S_i, i \in [m], j\in [n]\} \\
      E_7 &= \{(W_{i,p}, W_{i,q})|i\in [m], 2\leq p<q\leq 5\}\\
      E_8 &= \{(Z_{j,p}, Z_{j,q})|j\in [n], 2\leq p<q\leq 5\}\\
      \LL &= \{Z_{j,p}| j\in [n], p=2,3,4,5\} \\
      b &= t, d = 4m+n(t+1)+t
 \end{align*}

 We now claim that these two instances are equivalent. In one direction, let us assume that the \SC instance is a \YES instance. By renaming, let us assume that the collection $\{S_1, \ldots, S_t\}$ forms a valid set cover of the instance. We add the edges in the set $\EE^\pr=\{(X_{i,1}, X_{i,2}): i\in[t]\}$ in the graph \GG. Let the resulting graph be \HH. We claim that the core centrality of every vertex in $\{Z_{i,1}: i\in[n]\}\cup\{X_{i,j}:i\in[t],j\in[n]\}\cup\{W_{i,1}:i\in[t]\}\cup\{W_{i,j}:i\in[m], j\in[4]\}$ is $3$ in \HH. We first observe that the core centrality of every leader remains $3$ even after adding the edges in $\EE^\pr$. Also, for any $i\in[m]$, if the edge $(X_{i,1},X_{i,2})$ is added in the graph, the core centrality of the $n+1$ vertices $X_{i,t}, t\in [n]$ and $W_{i,1}$ become $3$. Hence after addition of the edges in $\EE^\pr$ in \GG, the core centrality of every vertex in $\{X_{i,j}:i\in[t],j\in[n]\}\cup\{W_{i,1}:i\in[t]\}$ becomes $3$. Since $\{S_1, \ldots, S_t\}$ forms a set cover for \UU, the core centrality of every vertex in $\{Z_{i,1}: i\in[n]\}$ becomes $3$. Lastly, the core centrality of every vertex in $\{W_{i,j}:i\in[m], j\in[4]\}$ was already $3$ in \GG and since addition of edges never decreases the core centrality of any vertex, the core centrality of these vertices are at least $3$ in \HH. Hence the \HL instance is a \YES instance.
 
 For the other direction, let us assume that there exists a set $\EE^\pr$ of edges such that in the graph $\HH=(\VV[\GG],\EE[\GG]\cup\EE^\pr)$, the core centrality of at least $d$ vertices in $\VV[\GG]\setminus \LL$ is at least $3$; let the set of followers with core centrality at least $3$ in \HH be $\YY\subseteq\VV[\GG]\setminus \LL$. Since adding edges in the graph never decreases the core centrality of any vertex, we have $\{W_{i,j}:i\in[m], j\in[4]\}\subseteq\YY$. Let us consider the following subset $\JJ\subseteq[m]$ defined as: $\JJ=\{j\in[m]:\exists 1\le i<k\le n \text{ with } (X_{j,i},X_{j,k})\in\FF\}$. Since $|\FF|\le b$ and $b=t$, we have $|\JJ|\le t$. We claim that $\{S_j: j\in\JJ\}$ forms a set cover for \UU. Suppose not, then at most $n-1$ vertices in $\{Z_{i,1}:i\in[n]\}$ can belong to \YY since $Z_{\el,1}$ does not belong to \YY if $u_\el$ is uncovered. Also, any vertex in $\{X_{i,\el}, W_{i,1}: i\in[m]\setminus\JJ,\el\in[n]\}$ does not belong to \YY. Hence, we have $|\YY|\le 4m+t(n+1)+n-1<d$ which contradicts our assumption that \FF forms a valid solution for the \HL instance. Hence the \SC instance is a \YES instance.
\end{proof}

\begin{figure}[t]
%\vspace{-3mm}
    \centering
    {\includegraphics[width=0.42\textwidth]{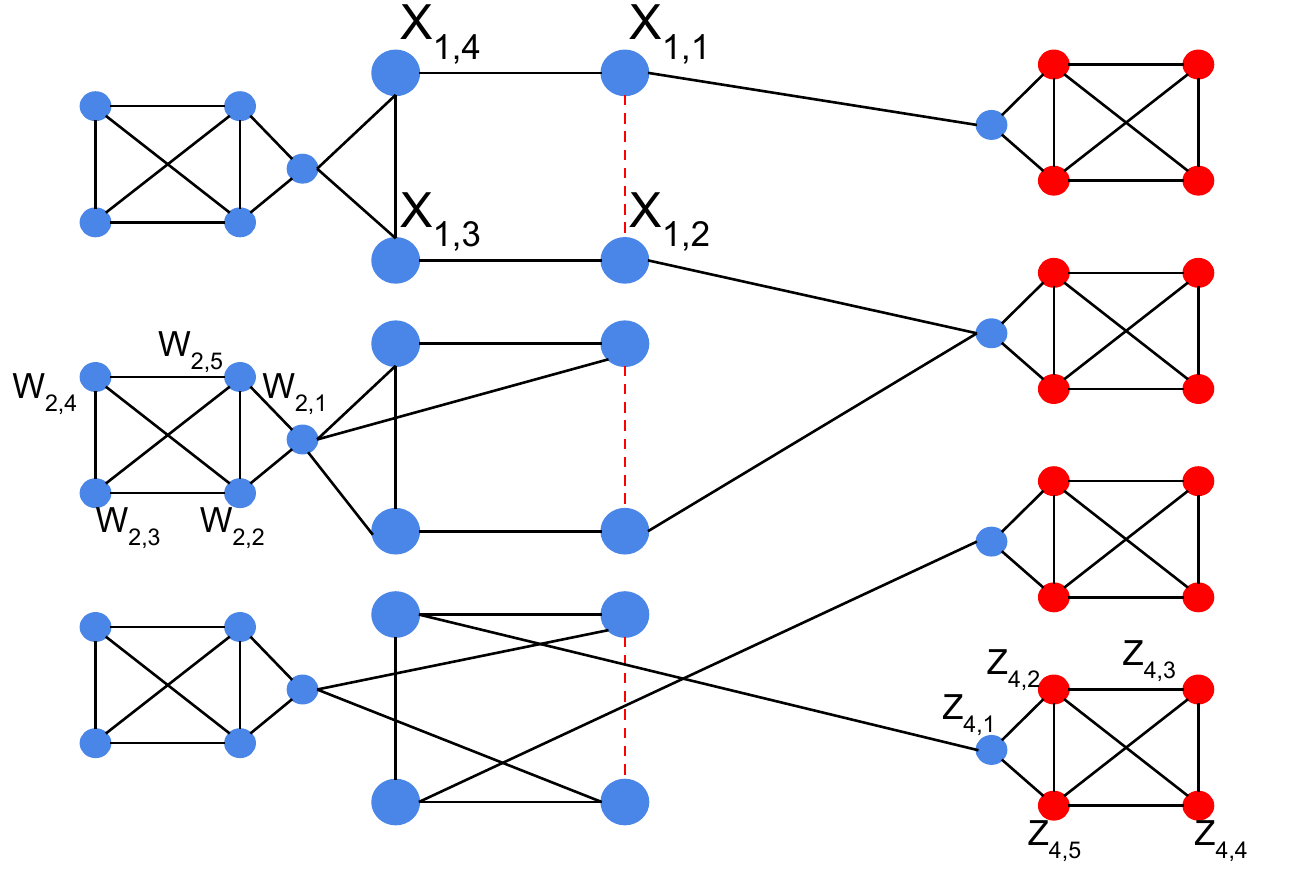}}
    \caption{\textbf{Example construction for hardness from Set Cover where  $\UU=\{u_1,u_2,u_3,u_4\}, S=\{S_1,S_2,S_3\}, S_1=\{u_1,u_2\},S_2=\{u_2\},S_3=\{u_3,u_4\}$. The red nodes are the leaders and the blue nodes are the followers. \label{fig:hardness_ex1}}}
    \vspace{-2mm}
\end{figure}

\Cref{thm:CHL_hard} along with well known inapproximability result for the \SC problem immediately give us the following result.

\begin{corollary}\label{cor:core_inapp}
 There does not exists any polynomial time algorithm for approximating the number of edges one needs to add in the \HL problem for core centrality within an approximation ratio of $(1-\alpha)\ln n$ for any constant $\alpha$ assuming $\Pb\ne\NP$ even when the core centrality of every leader is $3$.
\end{corollary}

\begin{proof}
 The result follows from the observation that the reduction in the proof of \Cref{thm:CHL_hard} is approximation preserving and the $(1-\alpha)\ln n$ inapproximability result for the \SC problem for any constant $\alpha$ assuming $\Pb\ne\NP$~\cite{Moshkovitz15}.
\end{proof}

% We now show that the hardness result in \Cref{thm:CHL_hard} is tight in the sense that if the core centrality of every leader is at most $2$ in the network, then the corresponding \HL problem is polynomial time solvable.

We now show that the hardness result in \Cref{thm:CHL_hard} is almost tight in the sense that if the core centrality of every leader is at most $1$ in the network, then the corresponding \HL problem is polynomial time solvable.

\begin{proposition}\label{thm:core_poly}
 There exists a polynomial time algorithm for the \HL problem for core centrality if the core centrality of every leader in the network is at most $1$.
\end{proposition}

\begin{proof}
%  One of the following two cases hold if the core centrality of every leader in the network is at most $2$.
 
%  {\em Case 1: The core core centrality of every leader is at most $1$.} 
 We observe that if the degree of any vertex $x$ is at least $1$, then its core centrality is at least $1$. Let $\FF^\pr\subseteq\FF$ be the subset of followers whose core centrality is at least $1$; say $|\FF^\pr|=d^\pr$. Hence the degree of every vertex in $\FF\setminus\FF^\pr$ is $0$. We add $\lceil \nfrac{(d-d^\pr)}{2}\rceil$ new edges such that the degree of at least $d-d^\pr$ vertices in $\FF\setminus\FF^\pr$ becomes at least $1$ in the resulting graph. We output \YES if $\lceil \nfrac{(d-d^\pr)}{2}\rceil\ge b$; otherwise we output \NO. Since any optimal solution must add at least $\lceil \nfrac{(d-d^\pr)}{2}\rceil$ edges, our algorithm is correct.
  
%  {\em Case 2: The core centrality of every leader is at most $2$ and there exists a leader with core centrality $2$.} First we make the following assumptions: the graph is connected, there are no isolated vertices in the graph, and all the leaders have core centrality $2$. However the proof can be extended with simple modifications without these assumptions. Let there already exists a set $F^\pr\subseteq F$ of $d^\pr$ followers in $F$ whose core centrality is $2$ (note, $k^*=2$). Let $f_i\in F, i\in[d-d^\pr]$ be the $(d-d^\pr)$ followers that are in $1$-core among the vertices in $F\setminus F^\pr$. These followers will form a forest (a set of trees) in the graph. We find the longest chain from all the trees hanging either from the followers or the leaders with core centrality $2$. This can be done in polynomial time. We add an edge between two end points of the two longest chains. Thus all the vertices in these two chains will form a cycle and will be in $2$-core. Note that this is a greedy algorithm, however, the gain of adding edges are independent of each other. Thus, the objective function is linear and this greedy approach will be optimal.
\end{proof}

\section{Captain Networks}
\label{sec:captain_nw}

In this section, we show the ``captain network'', originally proposed by Waniek et al.~\cite{DBLP:conf/atal/WaniekMRW17}, also ensures that the core centrality of any leader is at most the core centrality of any captain. They propose two constructions; one for single leader and another for multiple leaders.

%one for single and another for multiple ones the number of leaders is more than one and another when we have only one leader.
% We discuss a few methods to design captain networks from scratch. Waniek et al.~\cite{DBLP:conf/atal/WaniekMRW17} show that how to build such networks for degree centrality and closeness centrality. Our results on designing a captain network are the followings:
% 
% \begin{itemize}
%     \item We show that the proposed constructions in ~\cite{DBLP:conf/atal/WaniekMRW17} gives same core centrality for both the leaders (for both single and multiple leaders) and the captains.
%     
%     \item We extend the previous construction to increase the effect of disguise for the leaders, i.e., larger core centrality for the captains than the leaders.
%     
%     \item We prove that in all the above constructions, the difference between the leaders' influence after and before the edge additions among captains is non-negative. Thus their influence in captain networks remain non-decreasing.
% \end{itemize}

\subsection{For Multiple Leaders} 
We first describe the construction in~\cite{DBLP:conf/atal/WaniekMRW17}. % (let us call it as construction ML). The constructed network is called ``captain'' network. We prove that the core centrality of the leaders and the captains (imply $d$ followers) are the same.
The set \LL of leaders forms a clique. Each leader $l_i\in \LL$ has a corresponding group of $p$ captains $\CC_i=\{C_{i,1}, C_{i,2},\cdots, C_{i,p}\}$ and $l_i$ is connected to all vertices in $\CC_i$. Assuming that $|\LL|=h \geq 2$, there are $h$ such sets of captains $\{\CC_1,\CC_2,\cdots,\CC_h\}$. All vertices in the captain sets are connected as a complete $h$-partite graph. A captain $C_{i,j}$ serves two things: 1) It helps to hide the leader by being higher or of same centrality than the leader with maximum centrality. 2) It spreads the influence from the leader to the rest of the network. The remaining vertices $\XX=\{X_1,X_2,\cdots, X_m\}$ are each connected to one captain from each group $C_i$. The follower set in the network is $\FF=\XX\cup \CC_1\cup \CC_2\cup\cdots\cup\CC_h$. Let us call the resulting graph $\GG_M$. We now show that the core centrality of every leader in $\GG_M$ is at most the core centrality of every captain.

\begin{theorem} \label{thm:core_captain_mult}
Given a captain network $\GG_M$, let $r=\lfloor \frac{m}{p}\rfloor$ denotes the minimum number of connections that a captain $C_{i,j}$ has with vertices from $X$. Assuming we have at least $2$ leaders and $p>1$, the core centralities of the captains are either greater or same as the leaders. 
\end{theorem}

\begin{proof}
In $\GG_M$, the vertices in $\XX$ do not contribute in the core centrality of either the leaders or the captains. We observe that the degree of any vertex in $X$ is $h$. So their core centrality can be at most $h$. We claim that the captains and the leaders are in higher core than $h$. Consider a induced subgraph $G' \subset \GG_M$ that includes only the leaders, the captains and the edges between them. In $G'$, the degree of any captain $C_{i,j}$ is $p(h-1)+1$; on the other hand, the degree of any leader $l_i$ is $h-1+p$. So, in $G'$, all the captains and the leaders are at least in $d_{min}=min\{d(G,c_{i,j}),d(G,l_i)\}$-core. This comes from the fact that all the nodes (captains and leaders) have a minimum degree of $d_{min}$ and thus they are at least in $d_{min}$-core. Note that, $d(G,c_{i,j})- d(G,l_i)= p(h-1)+1-(h-1+p)=(h-2)(p-1)\geq 0$ as $h\geq 2$ and $p>1$. This implies the captains have higher degree than the leaders in $G'$. So, the captains have at least the same core-centrality as the leaders. Our claim is proved. Additionally, any vertex in $X$ is in $h$-core and $h<d_{min}$ assuming $p>1$. 
\end{proof}
\begin{corollary} 
\label{cor:core_hide_leader}
Given the same captain network $G$, assuming $h> 2$ and $p>1$, the core centrality of all the captains is strictly larger than the leaders. 
\end{corollary}

\begin{proof}
The key idea is that the captains will form a core only among themselves and that will be higher core than the leaders. Now, for any captain $c_{i,j}$ the degree among themselves is $p(h-1)$. Now that, $ p(h-1)-(h-1+p)=(h-2)(p-1) -1 > 0$ or $(h-2)(p-1) > 1$ is possible when $h> 2$ and $p> 1$. So, the captains have larger core-centrality than the leaders. 
\end{proof}

\subsection{For Single Leader}

We start with the construction in~\cite{DBLP:conf/atal/WaniekMRW17} when $h=1$ and show the core centralities of the leaders and the captains remain same in this case.

A single leader $l$ ($h=1$) has a corresponding two sets of $p$ captains $C_1=\{C_{1,1}, C_{1,2},\cdots, C_{1,p}\}$ and similarly it has $C_2$. All vertices in the captain sets are connected as a complete bipartite graph. Each of the remaining vertices in $X=\{x_1,x_2,\cdots, x_m\}$ is connected to one captain from each group $C_1$ and $C_2$. The follower set in the network is $F=\{X\cup C_1\cup C_2\}$. 

\begin{corollary} %\label{observation:core_captain_1}
Given the captain network described above, let $r=\lfloor \frac{m}{p}\rfloor$ denote the minimal number of connections that a captain, $c_{i,j}$ has with vertices from $X$. Assuming $h=1$, the core centralities of all the captains are same as the leader.
\end{corollary}

\begin{proof}
The proof follows from that of \Cref{thm:core_captain_mult}. The leader has degree $2p$ where as the captains have degree $1+p+r$. But the vertices in $X$ has only degree $2$. So the leader and the captains will be in the higher core and it will be $min\{2p,p+1\}$. Assuming $p\geq 2$ all the captain vertices and the leader will be in $p+1$-core. If $p=1$, all the vertices in the network will be in $2$-core.
\end{proof}

% Next we prove formally that the influence of the leader does not decrease before and after the addition of the edges among the captain vertices.
% 
% \begin{lemma} \label{lemma:influence}
% In IC model, the influence of the leader does not decrease after the addition of the edges among the captains in the captain network $G$ by construction ML.
% \end{lemma}
% 
% \begin{proof}
% The proof follows from the definition (Eq. \ref{eqn:influence}). When the edges get added between the captains, the number of paths from the leader to any captain gets increased. This is also true for the number of paths from any leader to any node in $X$. 
% \end{proof}
%\clearpage
\section{Simulation results}
In this section, we evaluate the performance of our $2$ approximation algorithm in \Cref{thm:2apx_opt_fol_deg} using synthetic networks. For brevity, let us call our algorithm in \Cref{thm:2apx_opt_fol_deg} as HLDA and called the lower bound used in \Cref{thm:2apx_opt_fol_deg} as LB. We also show how well the leaders can be hidden in the captain network via the core centrality measure. Solutions were implemented in Java and experiments conducted on $3.30$ GHz Intel cores with $30$ GB RAM.

\subsection{Evaluation of 2-Approximation Algorithm in \Cref{thm:2apx_opt_fol_deg}}
\begin{figure}[t]
%\small
 %\vspace{-6mm}
  \vspace{-1mm}
    \centering
    \subfloat[BA (avg. degree = 2) ]{\includegraphics[width=0.22\textwidth]{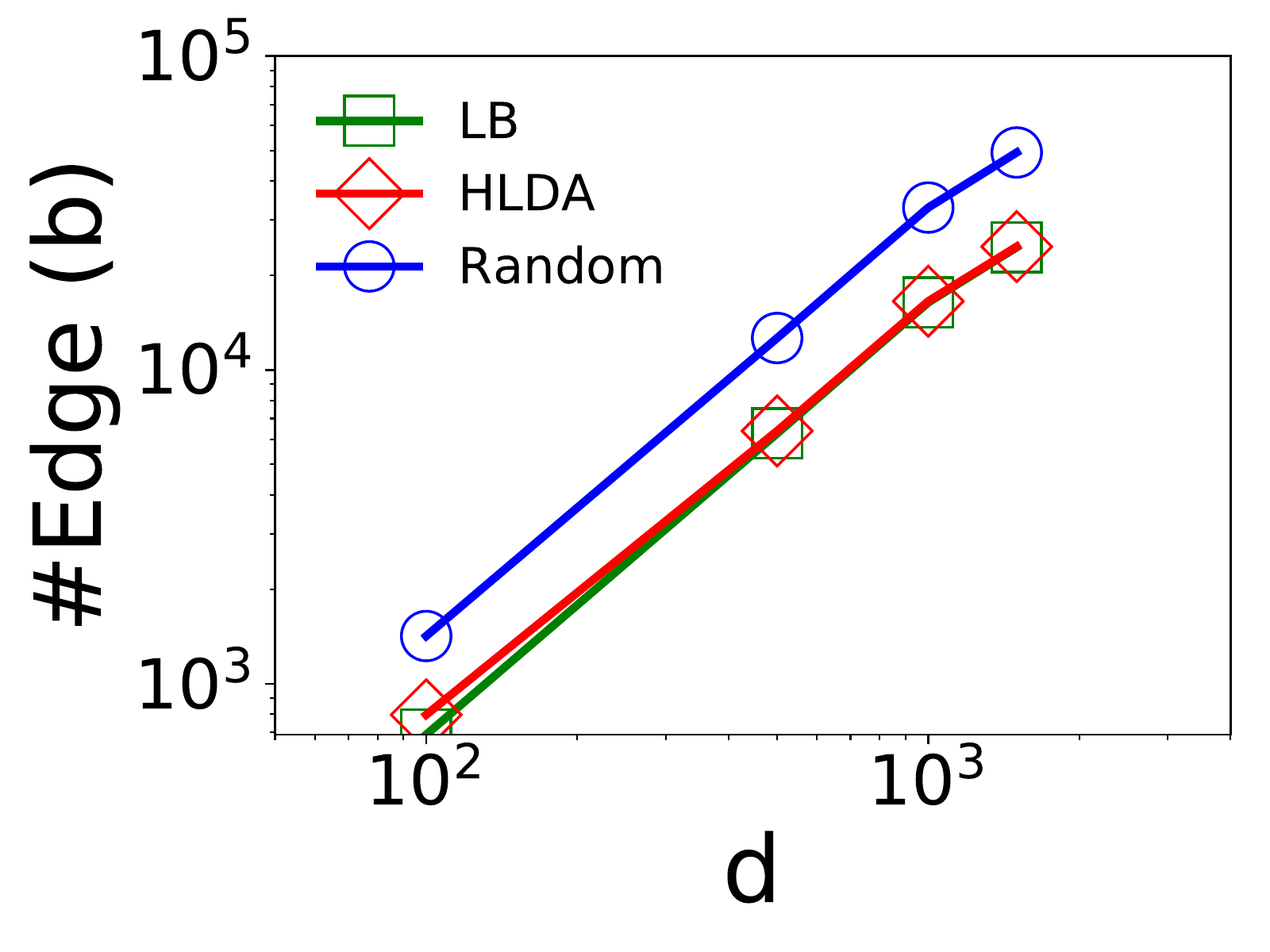}\label{fig:BA_2}}
     \subfloat[WS (avg. degree = 2)]{\includegraphics[width=0.22\textwidth]{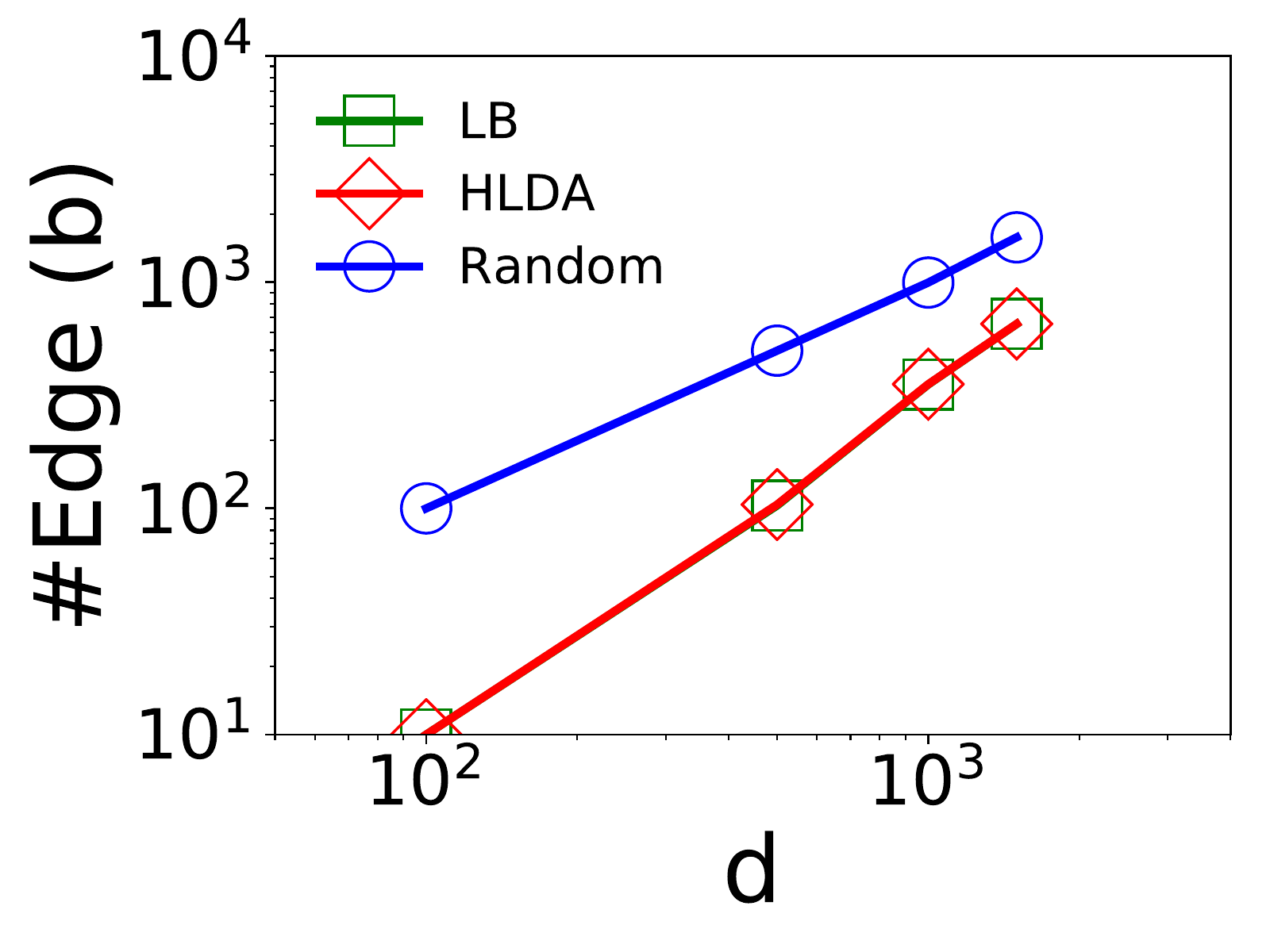}\label{fig:WS_2}} \\
    \subfloat[BA (avg. degree = 4)]{\includegraphics[width=0.22\textwidth]{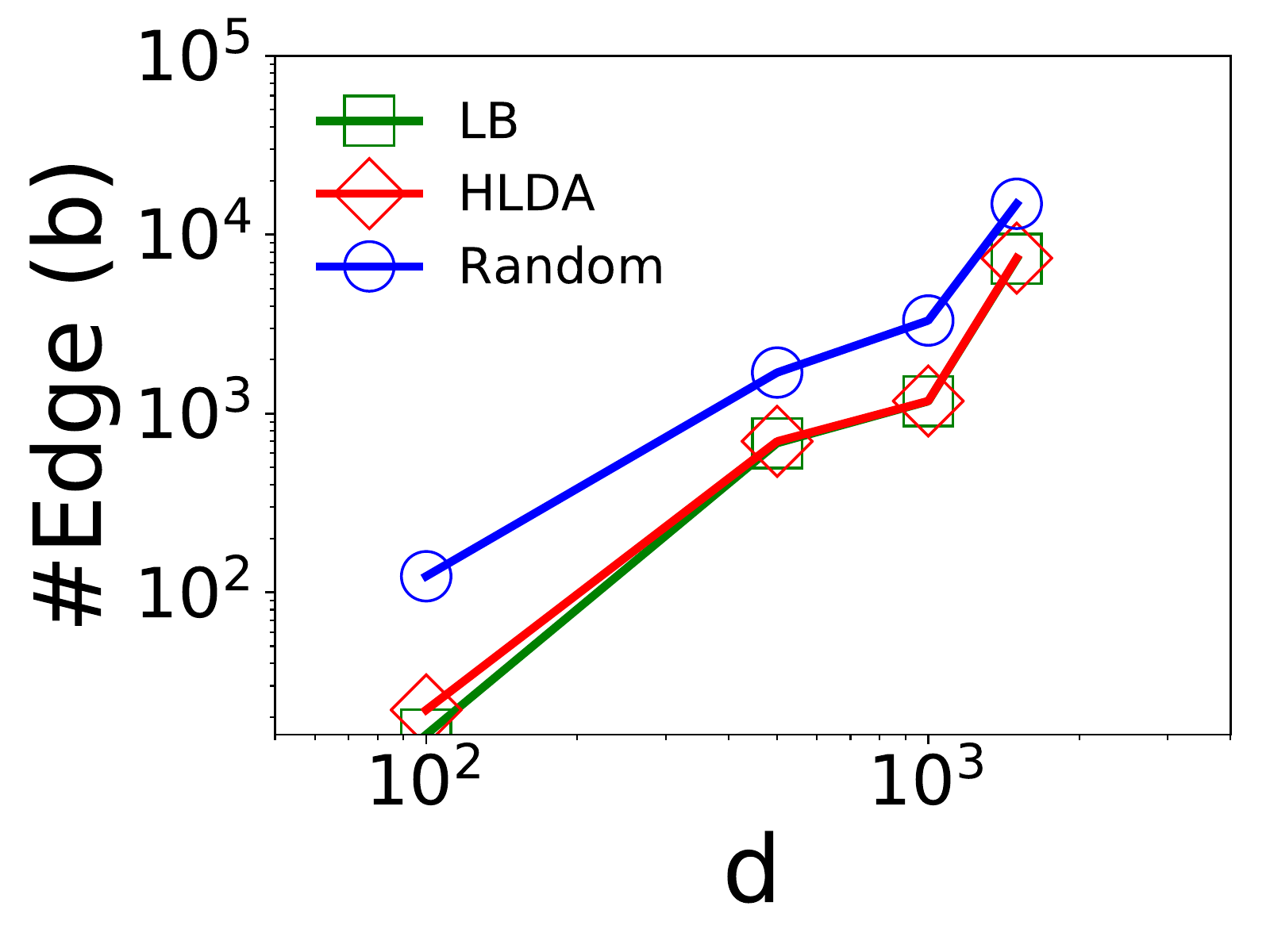}\label{fig:BA_4}} 
    \subfloat[WS (avg. degree = 4)]{\includegraphics[width=0.22\textwidth]{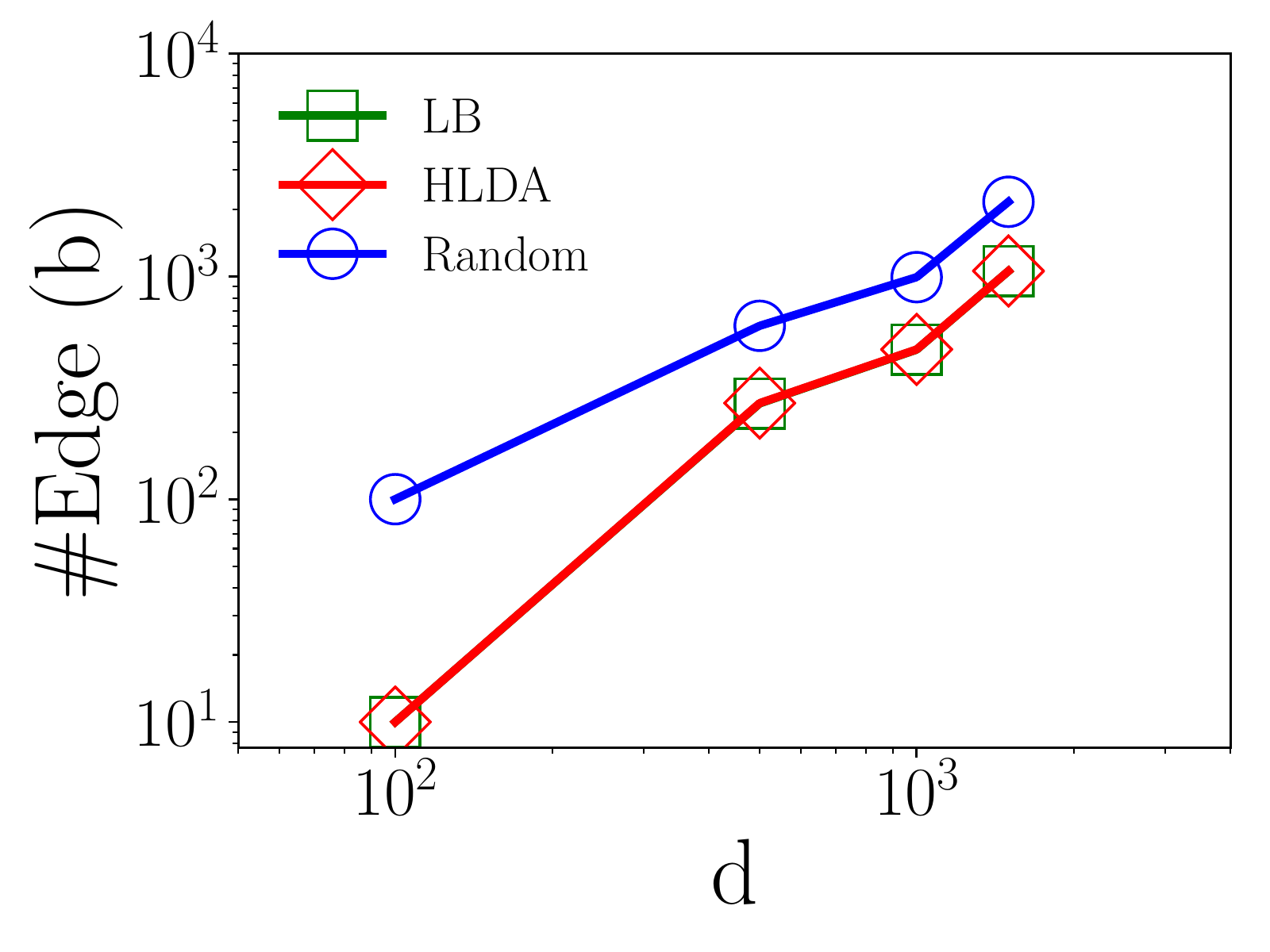}\label{fig:WS_4}} \\
     \subfloat[BA (avg. degree = 10)]{\includegraphics[width=0.22\textwidth]{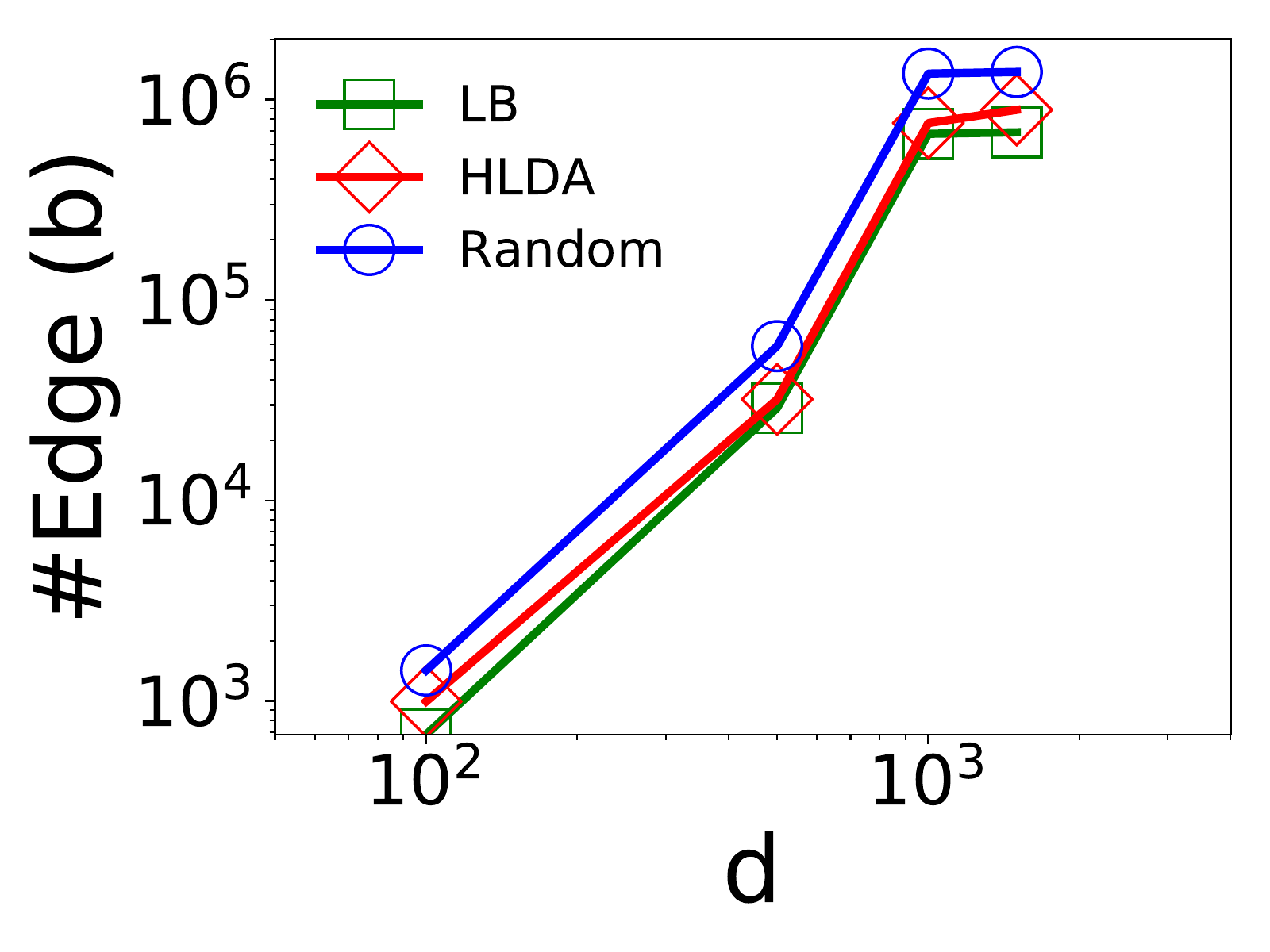}\label{fig:BA_10}} 
    \subfloat[WS (avg. degree = 10)]{\includegraphics[width=0.22\textwidth]{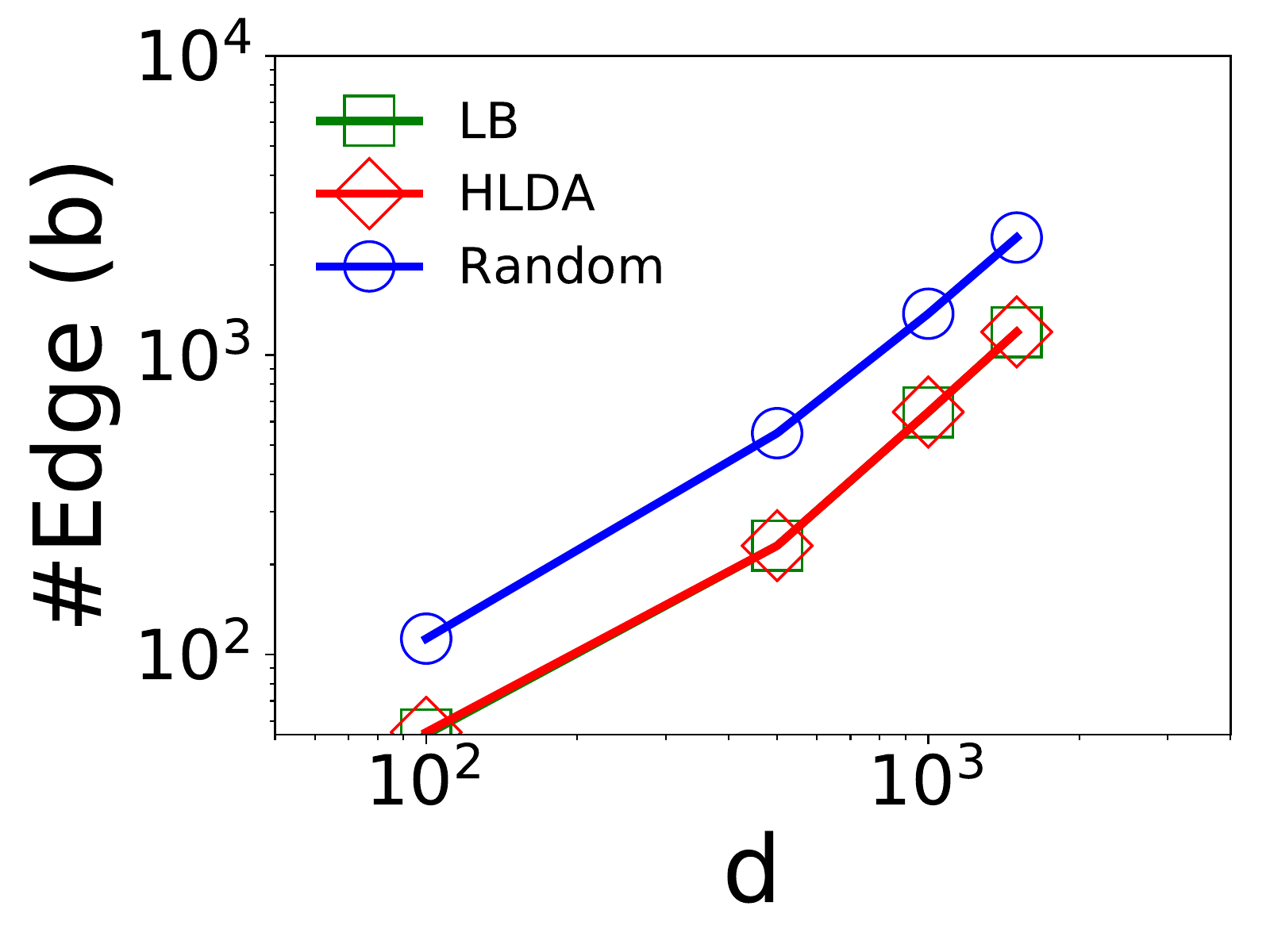}\label{fig:WS_10}}
    \vspace{-1mm}
    \caption{ Number of edges added ($b$) by different algorithms: LB implies a loose lower bound, HLDA is our algorithm that gives $2$-approximation and Random denotes a random edge addition algorithm. Clearly, in both networks while varying edge density (average degree of nodes), the number of edge addition by our algorithm HLDA is almost same as that of LB. \label{fig:baselines}}
    \vspace{-2mm}
\end{figure}

\textbf{Settings:} We generate synthetic network structures from two well-studied models: (a) Barabasi-Albert (BA)~\cite{barabasi1999emergence} and (b) Watts-Strogatz (WS)~\cite{watts1998collective}. While both have ``small-world'' property, WS do not have a scale-free degree distribution. We generate both the datasets of $70$ thousands vertices for three different edge densities: average degree of vertices as $2$, $4$ and $10$.  In the experiments we choose $20$ leaders ($|\LL|=h=20$) randomly from the top $100$ high degree vertices. 

\textbf{Baselines: }We compare our algorithm (HLDA) with two baselines. Our first baseline is the lower bound used in \Cref{thm:2apx_opt_fol_deg} which we call LB. Our second baseline is Random which denotes the number of random edges one needs to add to achieve the goal. The performance metric of the algorithms is the number of edges being added to satisfy the degree centrality requirement for $d$ followers. Hence, the quality is better when the number of edges is lower.  

\textbf{Results: } \Cref{thm:2apx_opt_fol_deg} shows that our algorithm (HLDA) proposed for degree centrality gives a $2$-approximation. However in practice it gives near optimal results. Figure \ref{fig:baselines} shows the results varying $d$ on four datasets. Note that, the axes are in logarithmic scale. In all six datasets, the number of solution edges of HLDA is similar to LB. However, Random cannot produce high quality results. Comparing the datasets (BA and WS), the algorithms (HLDA and LB) need higher number of edges in BA as the chosen leaders (randomly chosen from $100$ top degree nodes) have much higher degree than the followers due to the scale-free degree distribution.   

\begin{figure}[t]
%\small
 \vspace{-1mm}
    \centering
    \includegraphics[width=0.55\textwidth]{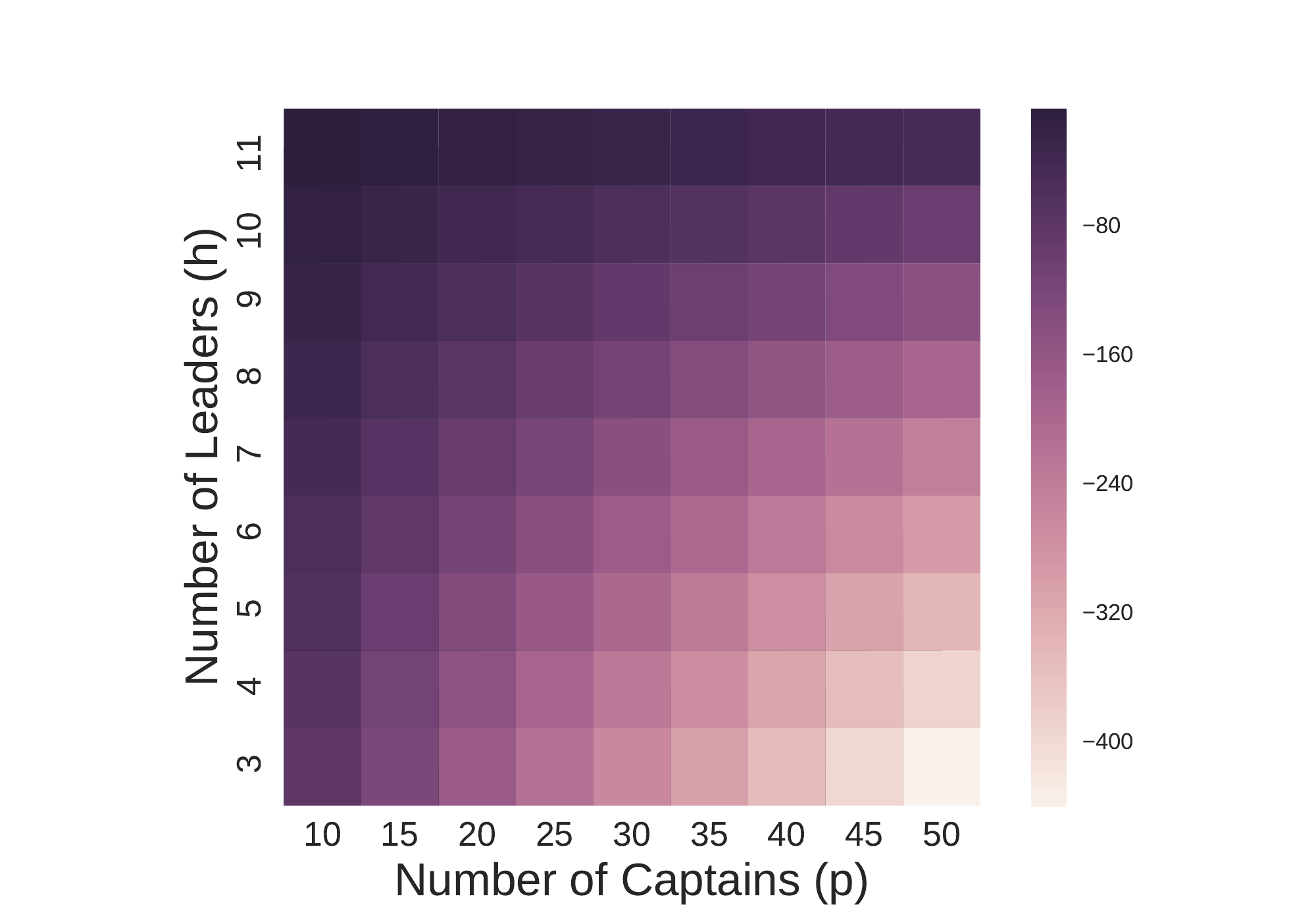}
    \vspace{-1mm}
    \caption{ Summary of the difference in core centralities between  a  leader  and  a captain in a given captain network (with $550$ vertices) by varying number of captains in each group  ($p$) and leaders ($h$).\label{fig:core}} %\label{fig:baselines}}
    %\vspace{-2mm}
\end{figure}

\subsection{Captain Networks and Core Centrality}
In section \ref{sec:captain_nw}, we prove that the core centrality of the leaders can be hidden by the captains in the captain networks \cite{DBLP:conf/atal/WaniekMRW17} (\Cref{thm:core_captain_mult} and Corollary \ref{cor:core_hide_leader}). We empirically evaluate the core centralities of the the leaders and the captains by varying two parameters: the number of captains ($p$) in each group $\CC_i$ and the number of leaders ($h$) for network with multiple leaders.% For every pair of the parameters, we measure the maximum difference in core centrality between (the greater the difference, the  greater  the  leaders’  disguise) any captain and any leader, and measured the influence of a leader to see how this influence is affected by the disguising process.

Figure \ref{fig:core} presents the results for a captain network with $550$ vertices. For every pair of two parameters ($p$ and $h$), we compute the maximum difference in core centrality between any leader and any captain. The intensity of the color signifies that the difference is lesser. Lesser difference also implies lesser disguise for the leaders. Low values of $p$ result into lower disguise for a leader. On the other hand, a high value of $p$ (large number of captains in each group) with low values of $h$ produces the maximum amount of disguise. But for a high value of $p$, if the number of leaders are also very high, i.e., high $h$, the amount of disguise for the leaders decreases.

We summarize our experimental findings as follows.
\begin{itemize}[itemsep=.2cm]
    \item HLDA produces near optimal results in practice, where as, Random cannot produce high quality results. HLDA and LB need more edges to satisfy the degree requirements for $d$ followers in BA due to the scale-free degree distribution. 
    \item A captain network with small number of leaders (low $h$) and a large number of captains in each group (high value of $p$) produces the maximum amount of disguise. 
    \item A low value of $p$, i.e., a small number of captains in each group yields lower disguise for core centrality which is not true for other centralities such as degree, closeness, and betweenness~\cite{DBLP:conf/atal/WaniekMRW17}.
\end{itemize}

\section{Conclusion and Future Work}
We have shown that the \HL problem for the core centrality measure is \NP-hard to approximate with a factor of $(1-\alpha)\ln n$ for any constant $\alpha>0$ for optimizing the number of edges one needs to add even when the core centrality of every leader is only $3$. On the other hand, we prove that the \HL leader problem for degree centrality is polynomial time solvable if the degree of every leader is $\OO(1)$. Moreover, we also provide a $2$ factor polynomial time approximation algorithm for the \HL problem for optimizing the number of edges one needs to add to hide all the leaders. Hence, our results prove that, although classical complexity theoretic framework fails to compare relative difficulty of hiding leaders with respect to various centrality measures~\cite{DBLP:conf/atal/WaniekMRW17}, hiding leaders may be significantly harder for the core centrality than the degree centrality. We complement our $2$ factor approximation algorithm for the \HL problem for degree centrality by proving that if there exists a $(2-\eps)$ factor approximation algorithm for the \HL problem for degree centrality for any constant $0<\eps<1$, then there exists a $\nfrac{\eps}{2}$ factor approximation algorithm for the famous Densest $k$-Subgraph problem which would be considered a major break through. The current best polynomial time algorithm for the Densest $k$-Subgraph problem achieves an approximation ratio of only $\tilde{\OO}(n^{1/4})$~\cite{DBLP:conf/soda/BhaskaraCVGZ12}. We have also empirically evaluated our approximation algorithm which shows that our algorithm produces an optimal solution for most of the cases. We have also shown that the captain networks proposed in \cite{DBLP:conf/atal/WaniekMRW17} can hide the leaders with respect to core centrality.

An important future direction is to explore the average case computational complexity of the \HL problem for popular network centrality measures. Since the results of Waniek et al.~\cite{DBLP:conf/atal/WaniekMRW17} and ours establish that the \HL problem is intractable only in the worst case, it could very well be possible that there exist heuristics that efficiently solve most randomly generated instances. If this is true, then the apparent complexity shield against manipulating various centrality measures will become substantially weak. Another immediate future work is to resolve the computational complexity of the \HL problem for the core centrality measure when the core centrality of every leader is at most $2$.

\subsection*{Acknowledgement}
Dey is funded by DST INSPIRE grant no. 04/2016/001479 and IIT Kharagpur grant no. IIT/SRIC/CS/VTS/2018-19/247.

% \clearpage
\bibliographystyle{alpha}
\balance
\bibliography{references}

\newcommand{\etalchar}[1]{$^{#1}$}
\begin{thebibliography}{WMWR18}

\bibitem[AADF17]{atran2017challenges}
Scott Atran, Robert Axelrod, Richard Davis, and Baruch Fischhoff.
\newblock Challenges in researching terrorism from the field.
\newblock {\em Science}, 355(6323):352--354, 2017.

\bibitem[AAE{\etalchar{+}}13]{altshuler2013stealing}
Yaniv Altshuler, Nadav Aharony, Yuval Elovici, Alex Pentland, and Manuel
  Cebrian.
\newblock Stealing reality: when criminals become data scientists (or vice
  versa).
\newblock In {\em Security and Privacy in Social Networks}, pages 133--151.
  Springer, 2013.

\bibitem[ABD15]{alimi2015dynamics}
Eitan~Y Alimi, Lorenzo Bosi, and Chares Demetriou.
\newblock {\em The dynamics of radicalization: a relational and comparative
  perspective}.
\newblock Oxford University Press, 2015.

\bibitem[Ant71]{anthonisse1971rush}
Jac~M Anthonisse.
\newblock The rush in a graph.
\newblock {\em Amsterdam: Mathematische Centrum}, 1971.

\bibitem[BA99]{barabasi1999emergence}
Albert-L{\'a}szl{\'o} Barab{\'a}si and R{\'e}ka Albert.
\newblock Emergence of scaling in random networks.
\newblock {\em Science}, 286(5439):509--512, 1999.

\bibitem[Bas69]{bass1969new}
Frank~M Bass.
\newblock A new product growth for model consumer durables.
\newblock {\em Manag. Sci.}, 15(5):215--227, 1969.

\bibitem[Bav48]{bavelas1948mathematical}
Alex Bavelas.
\newblock A mathematical model for group structures.
\newblock {\em Applied anthropology}, 7(3):16--30, 1948.

\bibitem[BCV{\etalchar{+}}12]{DBLP:conf/soda/BhaskaraCVGZ12}
Aditya Bhaskara, Moses Charikar, Aravindan Vijayaraghavan, Venkatesan
  Guruswami, and Yuan Zhou.
\newblock Polynomial integrality gaps for strong {SDP} relaxations of densest
  \emph{k}-subgraph.
\newblock In {\em Proc. 23-rd Annual {ACM-SIAM} Symposium on Discrete
  Algorithms, {SODA}}, pages 388--405, 2012.

\bibitem[Bea65]{beauchamp1965improved}
Murray~A Beauchamp.
\newblock An improved index of centrality.
\newblock {\em Behavioral science}, 10(2):161--163, 1965.

\bibitem[BF93]{baker1993social}
Wayne~E Baker and Robert~R Faulkner.
\newblock The social organization of conspiracy: Illegal networks in the heavy
  electrical equipment industry.
\newblock {\em Am. Sociol. Rev}, pages 837--860, 1993.

\bibitem[BFCB15]{belli2015exploring}
Roberta Belli, Joshua~D Freilich, Steven~M Chermak, and Katharine~A Boyd.
\newblock Exploring the crime--terror nexus in the united states: a social
  network analysis of a hezbollah network involved in trade diversion.
\newblock {\em Dynamics of Asymmetric Conflict}, 8(3):263--281, 2015.

\bibitem[BKL{\etalchar{+}}15]{bhawalkar2015preventing}
Kshipra Bhawalkar, Jon Kleinberg, Kevin Lewi, Tim Roughgarden, and Aneesh
  Sharma.
\newblock Preventing unraveling in social networks: the anchored k-core
  problem.
\newblock {\em SIAM J. Discrete Math.}, 29(3):1452--1475, 2015.

\bibitem[BKRW17]{DBLP:conf/soda/BravermanKRW17}
Mark Braverman, Young Kun{-}Ko, Aviad Rubinstein, and Omri Weinstein.
\newblock {ETH} hardness for densest-\emph{k}-subgraph with perfect
  completeness.
\newblock In {\em Proc. 28-th Annual ACM-SIAM Symposium on Discrete Algorithms
  ({SODA})}, pages 1326--1341, 2017.

\bibitem[Bra01]{brandes2001}
Ulrik Brandes.
\newblock A faster algorithm for betweenness centrality.
\newblock {\em Journal of mathematical sociology}, pages 163--177, 2001.

\bibitem[Cal17]{calvey2017covert}
David Calvey.
\newblock {\em Covert research: The art, politics and ethics of undercover
  fieldwork}.
\newblock Sage, 2017.

\bibitem[CAX{\etalchar{+}}05]{chen2005coplink}
Hsinchun Chen, Homa Atabakhsh, Jennifer~Jie Xu, Alan~Gang Wang, Byron Marshall,
  Siddharth Kaza, Lu~Chunju Tseng, Shauna Eggers, Hemanth Gowda, Tim Petersen,
  et~al.
\newblock Coplink center: social network analysis and identity deception
  detection for law enforcement and homeland security intelligence and security
  informatics: a crime data mining approach to developing border safe research.
\newblock In {\em Proc. 2005 National Conference on Digital government
  research}, pages 112--113. Digital Government Society of North America, 2005.

\bibitem[CDSV15]{crescenzi2015}
Pierluigi Crescenzi, Gianlorenzo D'Angelo, Lorenzo Severini, and Yllka Velaj.
\newblock Greedily improving our own centrality in a network.
\newblock In {\em Proc. 14th Symposium on Experimental Algorithms}, pages
  43--55, 2015.

\bibitem[CEHS12]{crossley2012covert}
Nick Crossley, Gemma Edwards, Ellen Harries, and Rachel Stevenson.
\newblock Covert social movement networks and the secrecy-efficiency trade off:
  The case of the uk suffragettes (1906--1914).
\newblock {\em Soc. Networks}, 34(4):634--644, 2012.

\bibitem[CLWU13]{csermely2013structure}
Peter Csermely, Andr{\'a}s London, Ling-Yun Wu, and Brian Uzzi.
\newblock Structure and dynamics of core/periphery networks.
\newblock {\em Journal of Complex Networks}, 1(2):93--123, 2013.

\bibitem[CSW05]{carrington2005models}
Peter~J Carrington, John Scott, and Stanley Wasserman.
\newblock {\em Models and methods in social network analysis}, volume~28.
\newblock Cambridge university press, 2005.

\bibitem[DK12]{demiroz2012anatomy}
Fatih Demiroz and Naim Kapucu.
\newblock Anatomy of a dark network: the case of the turkish ergenekon
  terrorist organization.
\newblock {\em Trends in organized crime}, 15(4):271--295, 2012.

\bibitem[DKN19]{DBLP:conf/comad/DeyKN19}
Palash Dey, Pravesh~K. Kothari, and Swaprava Nath.
\newblock The social network effect on surprise in elections.
\newblock In {\em Pro. {ACM} India Joint International Conference on Data
  Science and Management of Data, {24th COMAD, 6th CODS}, Kolkata, India,
  January 3-5, 2019}, pages 1--9, 2019.

\bibitem[DKS14]{duijn2014relative}
Paul~AC Duijn, Victor Kashirin, and Peter~MA Sloot.
\newblock The relative ineffectiveness of criminal network disruption.
\newblock {\em Scientific reports}, 4:4238, 2014.

\bibitem[DLG11]{dilkina2011}
Bistra Dilkina, Katherine~J. Lai, and Carla~P. Gomes.
\newblock Upgrading shortest paths in networks.
\newblock In {\em Integration of AI and OR Techniques in Constraint Programming
  for Combinatorial Optimization Problems}, pages 76--91. Springer, 2011.

\bibitem[DSV16]{DANGELO2016153}
Gianlorenzo D'Angelo, Lorenzo Severini, and Yllka Velaj.
\newblock On the maximum betweenness improvement problem.
\newblock {\em Electronic Notes in TCS}, 322:153 -- 168, 2016.

\bibitem[DZ10]{demaine2010}
E.~D. Demaine and M.~Zadimoghaddam.
\newblock Minimizing the diameter of a network using shortcut edges.
\newblock {\em in SWAT, ser.Lecture Notes in Computer Science, H. Kaplan,Ed.},
  pages 420--431, 2010.

\bibitem[Eis18]{eiselt2018destabilization}
HA~Eiselt.
\newblock Destabilization of terrorist networks.
\newblock {\em Chaos, Solitons \& Fractals}, 108:111--118, 2018.

\bibitem[EJ10]{enders2010network}
Walter Enders and Paan Jindapon.
\newblock Network externalities and the structure of terror networks.
\newblock {\em J. Confl. Resolut.}, 54(2):262--280, 2010.

\bibitem[ES07]{enders2007rational}
Walter Enders and Xuejuan Su.
\newblock Rational terrorists and optimal network structure.
\newblock {\em J. Confl. Resolut.}, 51(1):33--57, 2007.

\bibitem[FKB17]{farooq2017covert}
Ejaz Farooq, Shoab~A Khan, and Wasi~Haider Butt.
\newblock Covert network analysis to detect key players using correlation and
  social network analysis.
\newblock In {\em Proc. 2nd International Conference on Internet of things and
  Cloud Computing}, pages 94:1--94:6. ACM, 2017.

\bibitem[Fre77]{freeman1977set}
Linton~C Freeman.
\newblock A set of measures of centrality based on betweenness.
\newblock {\em Sociometry}, pages 35--41, 1977.

\bibitem[GJ79]{garey1979computers}
Michael~R Garey and David~S Johnson.
\newblock {\em Computers and {I}ntractability}, volume 174.
\newblock freeman New York, 1979.

\bibitem[GLM01]{goldenberg2001using}
Jacob Goldenberg, Barak Libai, and Eitan Muller.
\newblock Using complex systems analysis to advance marketing theory
  development: Modeling heterogeneity effects on new product growth through
  stochastic cellular automata.
\newblock {\em J. Acad. Mark. Sci.}, 9(3):1--18, 2001.

\bibitem[IETB12]{ishakian2012framework}
Vatche Ishakian, D{\'o}ra Erdos, Evimaria Terzi, and Azer Bestavros.
\newblock A framework for the evaluation and management of network centrality.
\newblock In {\em Proc. SIAM International Conference on Data Mining}, pages
  427--438, 2012.

\bibitem[JM12]{janssen2012stable}
RHP Janssen and Herman Monsuur.
\newblock Stable network topologies using the notion of covering.
\newblock {\em Eur J Oper Res.}, 218(3):755--763, 2012.

\bibitem[JZV{\etalchar{+}}16]{johnson2016new}
Neil~F Johnson, M~Zheng, Yulia Vorobyeva, Andrew Gabriel, Hong Qi, Nicol{\'a}s
  Vel{\'a}squez, Pedro Manrique, Daniela Johnson, E~Restrepo, C~Song, et~al.
\newblock New online ecology of adversarial aggregates: Isis and beyond.
\newblock {\em Science}, 352(6292):1459--1463, 2016.

\bibitem[KGH{\etalchar{+}}10]{kitsak2010identification}
Maksim Kitsak, Lazaros~K Gallos, Shlomo Havlin, Fredrik Liljeros, Lev Muchnik,
  H~Eugene Stanley, and Hern{\'a}n~A Makse.
\newblock Identification of influential spreaders in complex networks.
\newblock {\em Nature physics}, 6(11):888, 2010.

\bibitem[Kil12]{kilberg2012basic}
Joshua Kilberg.
\newblock A basic model explaining terrorist group organizational structure.
\newblock {\em Studies in Conflict \& Terrorism}, 35(11):810--830, 2012.

\bibitem[KKT03]{kempe2003maximizing}
David Kempe, Jon Kleinberg, and {\'E}va Tardos.
\newblock Maximizing the spread of influence through a social network.
\newblock In {\em Proc. 9th ACM SIGKDD international conference on Knowledge
  discovery and data mining}, pages 137--146. ACM, 2003.

\bibitem[Kno15]{knoke2015emerging}
David Knoke.
\newblock Emerging trends in social network analysis of terrorism and
  counterterrorism.
\newblock {\em Emerging Trends in the Social and Behavioral Sciences: An
  Interdisciplinary, Searchable, and Linkable Resource}, pages 1--15, 2015.

\bibitem[Kre02]{krebs2002mapping}
Valdis~E Krebs.
\newblock Mapping networks of terrorist cells.
\newblock {\em Connections}, 24(3):43--52, 2002.

\bibitem[LBH09]{lindelauf2009influence}
Roy Lindelauf, Peter Borm, and Herbert Hamers.
\newblock The influence of secrecy on the communication structure of covert
  networks.
\newblock {\em Soc. Networks}, 31(2):126--137, 2009.

\bibitem[LLPC10]{lu2010social}
Yong Lu, Xin Luo, Michael Polgar, and Yuanyuan Cao.
\newblock Social network analysis of a criminal hacker community.
\newblock {\em J. Comp. Inf. Sys}, 51(2):31--41, 2010.

\bibitem[LM15]{lin2015}
Yimin Lin and Kyriakos Mouratidis.
\newblock Best upgrade plans for single and multiple source-destination pairs.
\newblock {\em GeoInformatica}, 19(2):365--404, 2015.

\bibitem[LT08]{liu2008towards}
Kun Liu and Evimaria Terzi.
\newblock Towards identity anonymization on graphs.
\newblock In {\em Proceedings of the 2008 ACM SIGMOD international conference
  on Management of data}, pages 93--106. ACM, 2008.

\bibitem[MBS18]{medya2018making}
Sourav Medya, Petko Bogdanov, and Ambuj Singh.
\newblock Making a small world smaller: Path optimization in networks.
\newblock {\em IEEE Transactions on Knowledge and Data Engineering},
  30(8):1533--1546, 2018.

\bibitem[MCU16]{mahmoodye2016}
Ahmad Mahmoody, E~Charalampos, and Eli Upfal.
\newblock Scalable betweenness centrality maximization via sampling.
\newblock In {\em 22nd ACM SIGKDD international conference on Knowledge
  discovery and data mining}, pages 1765--1773, 2016.

\bibitem[Mem12]{memon2012identifying}
Bisharat~Rasool Memon.
\newblock Identifying important nodes in weighted covert networks using
  generalized centrality measures.
\newblock In {\em Proc. European Intelligence and Security Informatics
  Conference}, pages 131--140. IEEE, 2012.

\bibitem[MGP07]{morselli2007efficiency}
Carlo Morselli, Cynthia Gigu{\`e}re, and Katia Petit.
\newblock The efficiency/security trade-off in criminal networks.
\newblock {\em Soc. Networks}, 29(1):143--153, 2007.

\bibitem[MI06]{meade2006modelling}
Nigel Meade and Towhidul Islam.
\newblock Modelling and forecasting the diffusion of innovation--a 25-year
  review.
\newblock {\em Int. J. Forecast}, 22(3):519--545, 2006.

\bibitem[Mos15]{Moshkovitz15}
Dana Moshkovitz.
\newblock The projection games conjecture and the np-hardness of ln
  n-approximating set-cover.
\newblock {\em Theory Comput.}, 11:221--235, 2015.

\bibitem[MSS{\etalchar{+}}18]{medya2018group}
Sourav Medya, Arlei Silva, Ambuj Singh, Prithwish Basu, and Ananthram Swami.
\newblock Group centrality maximization via network design.
\newblock In {\em Proc. 24th SIAM International Conference on Data Mining},
  pages 126--134. SIAM, 2018.

\bibitem[MT09]{meyerson2009}
Adam Meyerson and Brian Tagiku.
\newblock Minimizing average shortest path distances via shortcut edge
  addition.
\newblock In {\em Approximation, Randomization, and Combinatorial Optimization.
  Algorithms and Techniques (APPROX-RANDOM)}, pages 272--285. Springer, 2009.

\bibitem[MVRS18]{medya2018noticeable}
Sourav Medya, Jithin Vachery, Sayan Ranu, and Ambuj Singh.
\newblock Noticeable network delay minimization via node upgrades.
\newblock {\em Proceedings of the VLDB Endowment}, 11(9):988--1001, 2018.

\bibitem[OR02]{otte2002social}
Evelien Otte and Ronald Rousseau.
\newblock Social network analysis: a powerful strategy, also for the
  information sciences.
\newblock {\em J. Inf. Sci.}, 28(6):441--453, 2002.

\bibitem[PBG11]{papagelis2011}
Manos Papagelis, Francesco Bonchi, and Aristides Gionis.
\newblock Suggesting ghost edges for a smaller world.
\newblock In {\em International conference on Information and knowledge
  management (CIKM)}, pages 2305--2308, 2011.

\bibitem[PPT15]{parotisidis2015selecting}
N~Parotisidis, Evaggelia Pitoura, and Panayiotis Tsaparas.
\newblock Selecting shortcuts for a smaller world.
\newblock In {\em SIAM International Conference on Data Mining (SDM)}, pages
  28--36. SIAM, 2015.

\bibitem[PPT16]{parotsidis2016centrality}
Nikos Parotsidis, Evaggelia Pitoura, and Panayiotis Tsaparas.
\newblock Centrality-aware link recommendations.
\newblock In {\em Proc. 9th International ACM Conference on Web Search and Data
  Mining}, pages 503--512, 2016.

\bibitem[PS95]{paik1995}
D.~Paik and S.~Sahni.
\newblock Network upgrading problems.
\newblock {\em Networks}, pages 45--58, 1995.

\bibitem[RE16]{roberts2016monitoring}
Nancy Roberts and Sean Everton.
\newblock Monitoring and disrupting dark networks: A bias toward the center and
  what it costs us.
\newblock In {\em Eradicating Terrorism from the Middle East}, pages 29--42.
  Springer, 2016.

\bibitem[Res06]{ressler2006social}
Steve Ressler.
\newblock Social network analysis as an approach to combat terrorism: Past,
  present, and future research.
\newblock {\em Homeland Security Affairs}, 2(2), 2006.

\bibitem[RK14]{riondato2014}
Matteo Riondato and Evgenios~M Kornaropoulos.
\newblock Fast approximation of betweenness centrality through sampling.
\newblock In {\em Proc. 7th International ACM Conference on Web Search and Data
  Mining}, pages 413--422, 2014.

\bibitem[Sag04]{sageman2004understanding}
Marc Sageman.
\newblock {\em Understanding terror networks}.
\newblock University of Pennsylvania Press, 2004.

\bibitem[SC14]{stevenson2014change}
Rachel Stevenson and Nick Crossley.
\newblock Change in covert social movement networks: The ‘inner circle’of
  the provisional irish republican army.
\newblock {\em Soc. Mov. Stud.}, 13(1):70--91, 2014.

\bibitem[Sei83]{seidman1983network}
Stephen~B Seidman.
\newblock Network structure and minimum degree.
\newblock {\em Soc. networks}, 5(3):269--287, 1983.

\bibitem[Sha54]{shaw1954group}
Marvin~E Shaw.
\newblock Group structure and the behavior of individuals in small groups.
\newblock {\em The Journal of psychology}, 38(1):139--149, 1954.

\bibitem[SJ08]{shaikh2008network}
Muhammad~Akram Shaikh and Wang Jiaxin.
\newblock Network structure mining: locating and isolating core members in
  covert terrorist networks.
\newblock {\em WSEAS Transactions on Information Science and Applications},
  5(6):1011--1020, 2008.

\bibitem[Spa91]{sparrow1991application}
Malcolm~K Sparrow.
\newblock The application of network analysis to criminal intelligence: An
  assessment of the prospects.
\newblock {\em Soc. Networks}, 13(3):251--274, 1991.

\bibitem[SS02]{DBLP:conf/atal/SabaterS02}
Jordi Sabater and Carles Sierra.
\newblock Reputation and social network analysis in multi-agent systems.
\newblock In {\em Proc. 1st International Joint Conference on Autonomous Agents
  and Multiagent Systems, {AAMAS}}, pages 475--482, 2002.

\bibitem[VL15]{von2015organized}
Klaus Von~Lampe.
\newblock {\em Organized crime: analyzing illegal activities, criminal
  structures, and extra-legal governance}.
\newblock Sage Publications, 2015.

\bibitem[WCZM07]{wang2007social}
Fei-Yue Wang, Kathleen~M Carley, Daniel Zeng, and Wenji Mao.
\newblock Social computing: From social informatics to social intelligence.
\newblock {\em IEEE Intelligent systems}, 22(2), 2007.

\bibitem[WMRW17]{DBLP:conf/atal/WaniekMRW17}
Marcin Waniek, Tomasz~P. Michalak, Talal Rahwan, and Michael Wooldridge.
\newblock On the construction of covert networks.
\newblock In {\em Proc. 16th Conference on Autonomous Agents and MultiAgent
  Systems, {AAMAS}}, pages 1341--1349, 2017.

\bibitem[WMWR18]{waniek2018hiding}
Marcin Waniek, Tomasz~P Michalak, Michael~J Wooldridge, and Talal Rahwan.
\newblock Hiding individuals and communities in a social network.
\newblock {\em Nature Human Behaviour}, 2(2):139, 2018.

\bibitem[WS98]{watts1998collective}
Duncan~J Watts and Steven~H Strogatz.
\newblock Collective dynamics of ‘small-world’networks.
\newblock {\em Nature}, 393(6684):440, 1998.

\bibitem[XC05]{xu2005criminal}
Jennifer Xu and Hsinchun Chen.
\newblock Criminal network analysis and visualization.
\newblock {\em Commun. ACM}, 48(6):100--107, 2005.

\bibitem[Yos14]{yoshida2014}
Yuichi Yoshida.
\newblock Almost linear-time algorithms for adaptive betweenness centrality
  using hypergraph sketches.
\newblock In {\em Proc. 20th ACM SIGKDD international conference on Knowledge
  discovery and data mining}, pages 1416--1425, 2014.

\bibitem[ZG09]{Zheleva2009}
Elena Zheleva and Lise Getoor.
\newblock To join or not to join: The illusion of privacy in social networks
  with mixed public and private user profiles.
\newblock In {\em Proceedings of the 18th International Conference on World
  Wide Web}, pages 531--540. ACM, 2009.

\end{thebibliography}

\end{document}